%% file: main.tex
\documentclass[10,journal]{IEEEtran}
\usepackage{amsmath,amsthm,amsfonts,amssymb}
\usepackage{algorithmic}
\usepackage{algorithm}
\usepackage{array}

\usepackage{textcomp}
\usepackage{url}
\usepackage{verbatim}
\usepackage{graphicx}
\usepackage{cite}

\usepackage{subcaption}
\usepackage{bm}
\usepackage{makecell}
\usepackage{diagbox}
\usepackage{multirow}
\usepackage{pifont}
\usepackage{soul,color}
\usepackage{xurl}
\usepackage{enumitem}
\usepackage{booktabs}

\allowdisplaybreaks[4]

\theoremstyle{plain}
\newtheorem{thm}{Theorem}
\newtheorem{lem}[thm]{Lemma}

\newtheorem{cor}[thm]{Corollary}
\newtheorem{prop}{Proposition}
\newtheorem{rem}{Remark}
\newtheorem{sty1}{Theorem}
\newtheorem{defi}[sty1]{Definition}

\begin{document}


\title{Polarization Aware Movable Antenna}

\author{Runxin Zhang, Yulin Shao, Yonina C. Eldar
\thanks{R. Zhang and Y. Shao are with the Department of Electrical and Electronic Engineering, University of Hong Kong, Hong Kong S.A.R., China.
}
\thanks{Y. C. Eldar is with the Department of Mathematics and Computer Science, Weizmann Institute of Science, Rehovot 7610001, Israel.}
\thanks{Correspondence to: ylshao@hku.hk.}
}

\maketitle

\begin{abstract}
This paper presents a polarization-aware movable antenna (PAMA) framework that integrates polarization effects into the design and optimization of movable antennas (MAs). While MAs have proven effective at boosting wireless communication performance, existing studies primarily focus on phase variations caused by different propagation paths and leverage antenna movements to maximize channel gains. This narrow focus limits the full potential of MAs. In this work, we introduce a polarization-aware channel model rooted in electromagnetic theory, unveiling a defining advantage of MAs over other wireless technologies such as precoding: the ability to optimize polarization matching. This new understanding enables PAMA to extend the applicability of MAs beyond radio-frequency, multipath-rich scenarios to higher-frequency bands, such as mmWave, even with a single line-of-sight (LOS) path. Our findings demonstrate that incorporating polarization considerations into MAs significantly enhances efficiency, link reliability, and data throughput, paving the way for more robust and efficient future wireless networks.
\end{abstract}


\section{Introduction}
\label{SecI}
\input{SecI}

\section{System Model}
\label{SecII}

\input{SecII}

\section{The PAMA Framework}
\label{SecIII}
\input{SecIII}

\section{Impacts of Antenna Movements}
\label{SecIV}
\input{SecIV}

\section{Optimization with PAMA}
\label{SecV}
\input{SecV}

\section{Conclusion}\label{SecVI}
This paper presented the PAMA framework, which integrates polarization effects into the translation and rotation of antennas in wireless communication systems. By factoring in polarization, PAMA uncovers a defining advantage of MA: the ability to optimize polarization matching between the transmitter and receiver. This enhancement debunks the belief that MA is only effective in low-frequency multipath environments and extends the applicability of MA to higher frequency bands like the mmWave spectrum.

At mmWave frequencies, antennas are physically smaller, making them easier to move and rotate, which facilitates the practical implementation of PAMA systems. 
The dominance of LOS paths in the mmWave band simplifies channel estimation compared to traditional MA systems operating in environments with rich multipath. 
Additionally, PAMA achieves near-optimal performance with low rotation granularity. These factors collectively enhance the feasibility of deploying PAMA in real-world scenarios.

The insights from this work underscore the critical role of polarization in maximizing the performance of MA. Future research may focus on integrating PAMA with other advanced communication technologies, evaluating its performance in various high-frequency settings, and addressing practical implementation challenges such as antenna control mechanisms and mobility constraints. By advancing the understanding of how polarization can be harnessed in movable antenna systems, this work lays the foundation for developing more efficient and reliable wireless communication networks.

\appendices
\section{Electric Field Radiation}\label{sec:AppA}
\input{AppendixA.tex}

\section{Proof of Corollary \ref{cor:1}}\label{sec:AppB}
\input{AppendixB.tex}

\section{On the impact of estimation errors}\label{sec:AppC}
\input{AppendixC.tex}

\section{Comparison of optimization algorithms}\label{sec:AppD}
\input{AppendixD.tex}

\bibliographystyle{IEEEtran}
\bibliography{ref.bib}

\end{document}

%% file: SecI.tex



The relentless pursuit of high-speed and reliable wireless communication has become a cornerstone of modern technological advancement. Emerging applications such as autonomous vehicles, the artificial intelligence of things, and augmented reality demand seamless connectivity with unprecedented data rates and minimal latency\cite{dang2020should,wang2023road,shao2024theory}. These evolving requirements are pushing the boundaries of existing wireless communication systems, necessitating innovative solutions to address challenges like spectrum scarcity, signal degradation, and dynamic environmental conditions.

One promising technique that has emerged to meet these challenges is the use of movable antennas (MAs) \cite{zhu2023modeling,shao20246d,shi2024capacity,ning2024movable}, also known as fluid antenna systems (FAS) \cite{wong2020fluid,wong2021fluid,wong2022bruce,wong2020performance,khammassi2023new}. By dynamically adjusting the position and orientation of antennas, it is possible to achieve enhanced signal-to-noise ratio (SNR), wider coverage, and adaptability to changing environments and user locations, without solely relying on traditional methods like increasing transmission power, beamforming, or deploying additional infrastructure\cite{yoo2006optimality,tse2005fundamentals,tang2020wireless,wu2019intelligent,shao2021federated,zhu2023movable,ma2023mimo}. 

\begin{figure}
    \centering
    \includegraphics[width=0.8\linewidth]{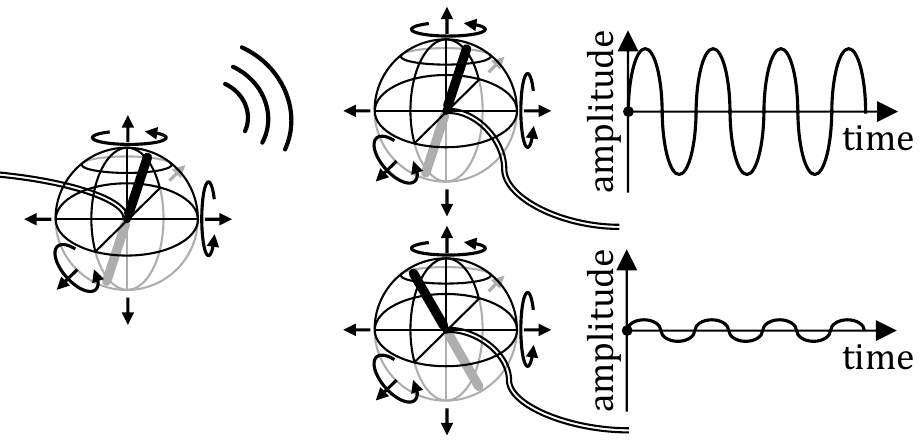}
    \caption{Movable antennas with the capability for 3D translation and rotation.}
    \label{f:demo}
\end{figure}

Previous research has explored the applications of MAs and FAS in radio frequency bands characterized by rich multipath environments \cite{zhu2023modeling,wong2020fluid}. In these settings, antenna movement is leveraged to exploit spatial diversity arising from amplitude fluctuations caused by multipath propagation.

In \cite{wong2020fluid}, the authors propose FAS that allows the receiving antenna to move to positions to capture the strongest signal. Analysis of outage probability demonstrates that, within a movement range of 0.2$\lambda$, FAS can achieve performance comparable to a maximal ratio combining system with five antennas spanning 2$\lambda$. 
This flexibility in antenna positioning offers significant advantages, especially in scenarios requiring massive connectivity. It allows each user to find an optimal position where instantaneous interference is minimized, effectively creating deep nulls. This spatial multiplexing gain facilitates multi-user access, a concept introduced as Fluid Antenna Multiple Access (FAMA) in \cite{wong2021fluid}.
Building upon this concept, \cite{wong2022bruce} extends FAS to a two-dimensional (2D) implementation and discusses reconfigurable pixel technologies as a practical approach for realization.
Additionally, \cite{psomas2023continuous} introduces a continuous FAS (CFAS) architecture, demonstrating that CFAS outperforms discrete FAS with limited ports, particularly at medium to high threshold levels.

The MA system presented in \cite{zhu2023modeling,zhu2023movable} exhibits additional phase variations due to antenna translation, surpassing fixed-position antennas in both deterministic and stochastic channels. The system's performance further improves with increased positioning accuracy, larger movement ranges, and enhanced multipath effects.
Building on MA system, advancements in six-dimensional (6D) systems, such as those discussed in \cite{shao20246d,zhu2024performance}, enable simultaneous three-dimensional (3D) translation and rotation, enhancing the system's directional capabilities. By imposing spatial constraints on surfaces equipped with 6D movable antennas (6DMA), user interference is reduced, improving the signal-to-interference-plus-noise ratio (SINR).
In \cite{shi2024capacity}, the 6DMA model is simplified by integrating conventional fixed-position antenna arrays with 6DMA surfaces capable of rotating along a circular path. By limiting the degrees of freedom for surface movement to the azimuth angle along the circular track, computational complexity is significantly reduced while enhancing the system's maximum capacity.

Despite these advancements, the impact of polarization changes due to antenna movement and rotation has been largely overlooked. Polarization is a fundamental characteristic of electromagnetic waves, describing the orientation of the electric field vector as the wave propagates through space \cite{stutzman2012antenna,costa2012doa,huang2021antennas,nabar2002performance,kurum2018scobi}. In wireless communication systems, proper alignment of polarization between transmitting and receiving antennas is crucial for efficient signal transmission and reception. When polarizations are properly aligned, the receiving antenna can maximally capture the transmitted energy, resulting in stronger signal strength, improved SNR, and enhanced overall system performance. Conversely, polarization mismatch occurs when there is a misalignment between the polarization orientations of the transmitting and receiving antennas, leading to significant signal attenuation as the misaligned component of the electromagnetic wave is effectively lost.
Neglecting polarization effects can lead to incomplete system models and suboptimal performance assessments, hindering a full understanding of the impact and advantages of MAs. 

This oversight is particularly critical as communication systems move towards higher frequency bands, such as the millimeter-wave (mmWave) spectrum \cite{niu2015survey,heath2016overview}. At these higher frequencies, signal propagation becomes more directional with reduced diffraction around obstacles, often resulting in a single line-of-sight (LOS) link between transmitter and receiver \cite{3GPP5G}. Previously, it was believed that MAs offer no gains in such LOS scenarios due to the lack of multipath diversity \cite{wong2020fluid,zhu2023modeling}. However, our formulation incorporating polarization effects indicates otherwise -- MAs can achieve significant performance gains even in LOS scenarios by facilitating polarization matching. 
As such, there is a clear gap in current research: the need to incorporate polarization effects into the design and optimization of MA systems, especially as the wireless industry advances towards higher frequency bands to achieve greater bandwidths and data rates. 

To fill the gap, this paper makes the following main contributions:
\begin{itemize}[leftmargin=0.5cm]
    \item We introduce a polarization-aware movable antenna (PAMA) framework, which incorporates polarization effects into the design and optimization of antenna movement and rotation. PAMA uncovers a previously overlooked, but defining benefit of antenna movement -- facilitating optimal polarization matching between the transmitter and receiver. By accounting for the impact of polarization changes due to movement and rotation, PAMA provides a more refined system model and deepens our understanding of the dynamic capabilities of MAs.
    \item Building upon the electromagnetic theory, we develop a polarization-aware channel model that describes the impact of 3D movement and rotation of antennas on the channel gains. Grounded in Maxwell's equations, this model allows for precise characterization of how antenna positioning and orientation affect polarization alignment and, consequently, signal reception quality. By integrating these factors into the system design, we enable the optimization of antenna movements to maintain polarization alignment, thereby maximizing the sum rate of the communication system.
    \item Contrary to prior beliefs that MAs are more effective in low-frequency, multipath-rich environments, our PAMA framework demonstrates that MAs are highly effective in higher frequency bands such as the mmWave band, even with a single LOS link. The smaller size of antennas at higher frequencies makes them easier to move and rotate, facilitating practical implementation of PAMA systems in these bands. This new understanding extends the potential and application of MA to next-generation wireless communication systems that demand higher bandwidths.
\end{itemize}

{\it Organization}: The remainder of this paper is structured as follows.
Section~\ref{SecII} introduces the system model and essential geometric concepts used throughout the work. Section~\ref{SecIII} presents our PAMA framework and the polarization-aware channel model.
Section~\ref{SecIV} explores the impacts of 3D antenna translation and rotation on the PAMA channel gains. Section~\ref{SecV} demonstrates the potential of PAMA through a sum-rate optimization problem. 
Finally, Section~\ref{SecVI} concludes the paper.

{\it Notations}: We use boldface lowercase letters to denote column vectors (e.g., $\bm{h}$, $\bm{w}$) and boldface uppercase letters to denote matrices (e.g., $\bm{H}$, $\bm{P}$).
For a vector or matrix, $(\cdot)^\top$ denotes the transpose, $(\cdot)^*$ denotes the complex conjugate, and $(\cdot)^H$ denotes the conjugate transpose.
$\mathbb{C}$ stand for the sets of complex values.
The imaginary unit is represented by $j$.
The Euclidean norm of a vector $\bm{p}$ is denoted by $|\bm{p}|$. For two vectors $\bm{a}$ and $\bm{b}$, $\langle \bm{a},\bm{b}\rangle$ denotes the angle between them, and $\bm{a}\cdot\bm{b}$ denotes their inner product.

%% file: SecII.tex
\begin{figure}
    \centering
    \includegraphics[width=0.8\linewidth]{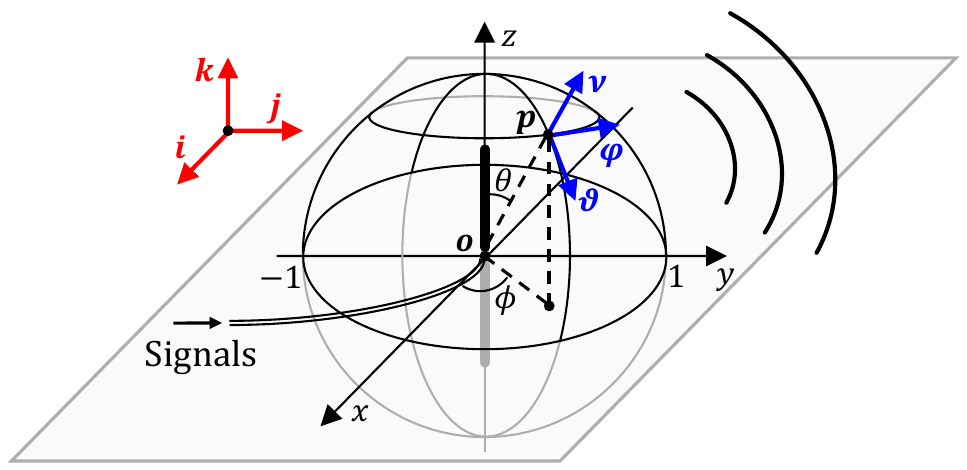}
    \caption{The two coordinate systems and their unit vectors.}
    \label{fig:CS}
\end{figure}

We consider a mobile edge network within a 3D environment. The network comprises a base station (BS) and $K$ mobile users distributed within the 3D space. The BS is equipped with $L$ half-wave dipole antennas, and each mobile user is equipped with a single half-wave dipole antenna.
We focus on the downlink transmission from the BS to the users, operating in the mmWave frequency band. Both the antennas at the BS and mobile users are movable.

\begin{defi}[Movable Antennas]
A movable antenna is capable of performing 3D translations and 3D rotations, allowing it to navigate through and orient within 3D space, as depicted in Fig.~\ref{f:demo}.
\end{defi}

Translation typically requires larger operational spaces, often spanning tens of wavelengths, whereas rotation can be accommodated within a more compact area, approximately the size of a half-wavelength sphere, as shown in Fig.~\ref{f:demo}. In this work, we assume that the antennas on mobile users are restricted to 3D rotation, due to the spatial constraint. In contrast, the $L$ antennas at the BS can perform both 3D rotation and translation independently within a designated region. The extension of this model to include translatable antennas for mobile users is straightforward.

In this work, we consider the far-field scenario, which implies the following:
\begin{itemize}[leftmargin=0.5cm]
\item {Plane wave.} The waves received by the users can be treated as plane waves because the propagation distance is much greater than the wavelength and the dimensions of the movable regions.
\item {Point representation of antennas.} The physical sizes of the transmitting and receiving antennas are negligible compared to the propagation distance, allowing us to represent their positions as points specified by their center coordinates.
\end{itemize}

To mathematically represent the network and the positions of the antennas, we establish two coordinate systems: the Cartesian Coordinate System (CCS) and the Spherical Coordinate System (SCS), as depicted in Fig.~\ref{fig:CS}.

\begin{defi}[Coordinate Systems]
We define a CCS to describe the positions of the transmitting and receiving antennas, with three orthogonal unit vectors denoted as $\bm{i}$, $\bm{j}$, and $\bm{k}$.

Correspondingly, an SCS is established with its origin at point $\bm{o} = (0, 0, 0)$ in the CCS, and the direction of $\bm{k}$ serving as the reference axis. The SCS has three unit vectors: the radial unit vector $\bm{\nu}$, the polar angle unit vector $\bm{\vartheta}$, and the azimuthal angle unit vector $\bm{\varphi}$.
\end{defi}

\begin{table}[t]
\centering
\caption{Conversion formulas between the unit vectors of CCS and SCS.}
\begin{tabular}{cccc}
\toprule
     \textbf{Dot products}   & $\bm{i}$ & $\bm{j}$ & $\bm{k}$ \\
     \midrule
     $\bm{\nu}$ & $\sin{\theta} \cos{\phi}$ & $\sin{\theta} \sin{\phi}$ & $\cos{\theta}$ \\
     $\bm{\vartheta}$ & $\cos{\theta} \cos{\phi}$ & $\cos{\theta} \sin{\phi}$ & $-\sin{\theta}$ \\
     $\bm{\varphi}$ & $-\sin{\theta}$ & $\cos{\phi}$ & $0$ \\
     \bottomrule
\end{tabular}
\label{tab:conversion}
\end{table}

Table \ref{tab:conversion} provides the conversion formulas between the unit vectors of the two coordinate systems. For instance, the radial unit vector $\bm{\nu}$ in the SCS can be expressed in terms of the CCS unit vectors as $\bm{\nu} = \sin(\theta) \cos(\phi) \bm{i} + \sin(\theta) \sin(\phi) \bm{j} + \cos(\theta) \bm{k}$.

Given the two coordinate systems, we introduce several definitions regarding the antennas' positions and orientations.

\begin{defi}[Positions of the Antennas]
In the far-field scenario, the position of an antenna can be represented by its center coordinates. We denote the center coordinate of the $\ell$-th transmitting antennas at the BS by $\bm{p}_{t,\ell} = \left(x_{t,\ell}, y_{t,\ell}, z_{t,\ell}\right)$, $\ell = 1, 2, \ldots, L$. Similarly, the center coordinate of the receiving antenna at the $\kappa$-th user is denoted by $\bm{p}_{r,\kappa} = \left(x_{r,\kappa}, y_{r,\kappa}, z_{r,\kappa}\right)$ for $\kappa = 1, 2, \ldots, K$.
\end{defi}

\begin{defi}[Directions of Antennas]
The direction of an antenna, characterized by a unit vector $\bm{n}$, is the orientation in which the antenna is pointed. We denote the direction for the $\ell$-th transmitting antenna at the BS by $\bm{n}_{t,\ell} = \big(x_{t,\ell}^{(n)}, y_{t,\ell}^{(n)}, z_{t,\ell}^{(n)}\big)$, and for the receiving antenna at the $\kappa$-th user by $\bm{n}_{r,\kappa} = \big(x_{r,\kappa}^{(n)}, y_{r,\kappa}^{(n)}, z_{r,\kappa}^{(n)}\big)$.
\end{defi}

\begin{defi}[Polar and Azimuthal Angles of Antennas]
The polar angle of an antenna, denoted by $\theta^{(n)} \in[0,\pi]$, is the angle between $\bm{k}$ and the direction of the antenna. The azimuthal angle of an antenna, denoted by $\phi^{(n)} \in[0, 2 \pi)$, is the angle measured counterclockwise from $\bm{i}$ to the projection of the antenna onto the $x\bm{o}y$ plane. 
For the $\ell$-th transmitting antenna, these angles are denoted by $\theta_{t,\ell}^{(n)}$ and $\phi_{t,\ell}^{(n)}$, respectively.
For the receiving antenna at the $\kappa$-th user, they are denoted by $\theta_{r,\kappa}^{(n)}$ and $\phi_{r,\kappa}^{(n)}$, respectively.
\end{defi}

\begin{figure}
    \centering
    \begin{subfigure}[!t]{.8\linewidth}
        \centering
        \includegraphics[width=\linewidth]{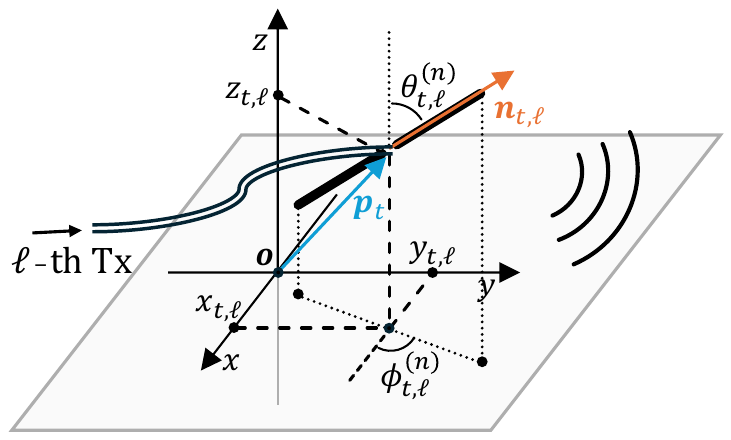}
        \caption{$\ell$-th transmitting antenna}\label{f:P2P_SNR}
        \includegraphics[width=\linewidth]{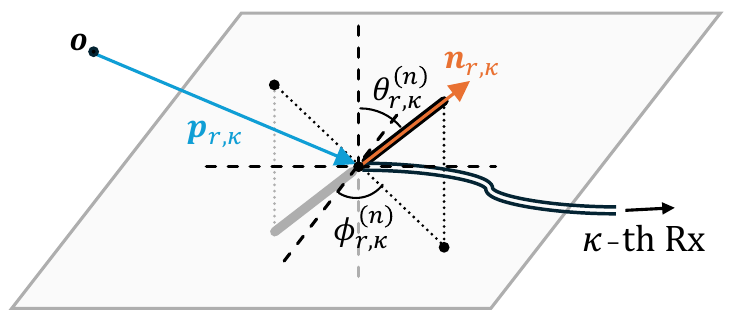}
        \caption{$\kappa$-th receiving antenna}\label{f:P2P_geo}
    \end{subfigure}
    \caption{Illustration of the polar and azimuthal angles, along with the geometric representation of the propagation link between the $\ell$-th transmitting antenna and the $\kappa$-th receiving antenna.}
    \label{f:geometry}
\end{figure}

The polar and azimuthal angles are illustrated in Fig.~\ref{f:geometry}. 
A summary of the definitions and their geometric relationships is given in Tab.~\ref{tab:notations}.

\begin{lem}
The direction of an antenna in the CCS can be represented by its polar and azimuthal angles. For the $\ell$-th transmitting antenna, we have
\begin{eqnarray} \label{e:nt}
&&\hspace{-1cm} \bm{n}_{t,\ell} = \left(x_{t,\ell}^{(n)}, y_{t,\ell}^{(n)}, z_{t,\ell}^{(n)}\right) \notag\\
&&\hspace{-0.3cm} = \left( \sin{\theta_{t,\ell}^{(n)}} \cos{\phi_{t,\ell}^{(n)}}, \sin{\theta_{t,\ell}^{(n)}} \sin{\phi_{t,\ell}^{(n)}}, \cos{\theta_{t,\ell}^{(n)}} \right).
\end{eqnarray}
\end{lem}

The CCS representation facilitates mathematical operations such as dot products, providing computational simplicity. In contrast, the SCS representation offers a clearer depiction of the relationship between channel gains and antenna orientation.

%% file: SecIII.tex
Building upon the system model, this section presents our PAMA framework. 
At the core of wireless communication is the interaction between transmitting and receiving antennas through electromagnetic fields. An alternating current in the transmitting antenna generates electromagnetic waves, creating an electromagnetic field that propagates through space to the receiver. The receiving antenna interacts with this field, specifically sensing the electric field component, which induces a current within the antenna used to recover the transmitted signal. This fundamental interaction is governed by Maxwell's equations, which describe how electric and magnetic fields are generated by charges and currents, and how they propagate in space.

In this section, we begin by modeling the radiation of the electric field from movable transmitting antennas, characterizing how movements -- 3D translations and rotations -- affect the spatial distribution and polarization of the emitted electromagnetic waves.
We then explore how movable receiving antennas detect and interact with the electric field. Ultimately, the analysis leads to a polarization-aware channel representation for movable antennas.

\subsection{Electric field radiation}
\label{Sec:SM_Tx}
Given that the electric field at the receiver is a superposition of waves from all transmitting antennas, we can focus on the radiation analysis of a single antenna. 
To facilitate the exposition, we introduce the emission angle of a transmitting antenna.

\begin{defi}[Emission Angle]
The emission angle $\theta^{(e)}\in[0,\pi]$ is the angle between the direction of the transmitting antenna and the propagation direction of the electromagnetic wave from the antenna to an observation point $\bm{p}$. The emission angle of the $\ell$-th transmitting antenna can be written as
\begin{equation}
\theta_{\ell,\bm{p}}^{(e)} \triangleq
\left\langle \bm{p} - \bm{p}_{t, \ell}, \bm{n}_{t,\ell} \right\rangle,
\end{equation}
where $\cos{\theta^{(e)}_{\ell,\bm{p}}} = \frac{\left(\bm{p} - \bm{p}_{t, \ell} \right) \cdot \bm{n}_{t,\ell}}{\left|\left(\bm{p} - \bm{p}_{t, \ell} \right)\right||\bm{n}_{t,\ell}|}$.
\end{defi}

Theorem \ref{thm:radiation} below summarizes the main result of this section.

\begin{thm}[Electric Field Radiation]
\label{thm:radiation}
Consider the $\ell$-th transmitting antenna located at $\bm{p}_{t,\ell}$ and oriented in the direction $\bm{n}_{t,\ell}$. Suppose the transmitted signal at $\ell$-th antenna is $s_\ell$. At any observation point $\bm{p}$ in space, the electric field induced by this antenna is given by
\begin{eqnarray} \label{e:El}
&&\hspace{-1cm} \bm{\mathcal{E}}_{\ell,\bm{p}} \triangleq \mathcal{E}_{\ell,\bm{p}} \bm{n}_{\bm{\mathcal{E}}_{\ell,\bm{p}}}, \\
&&\hspace{-1cm}  \mathcal{E}_{\ell,\bm{p}} =
2 j c \mu_1 s_\ell \frac{e^{-j \frac{2\pi }{\lambda}|\bm{p}|} }{4 \pi |\bm{p}| } \frac{\cos{\left( \frac{\pi}{2} \cos{\theta^{(e)}_{\ell,\bm{p}}} \right)}}{\sin{\theta^{(e)}_{\ell,\bm{p}}}} e^{j \frac{2\pi}{\lambda} \frac{\bm{p} \cdot \bm{p}_{t,\ell}}{|\bm{p}|}},
\notag\\
&&\hspace{-1cm}  
\bm{n}_{\bm{\mathcal{E}}_{\ell,\bm{p}}} = \frac{\bm{n}_{t, \ell} - \left( \bm{n}_{t, \ell} \cdot \frac{\bm{p}}{|\bm{p}|} \right)  \frac{\bm{p}}{|\bm{p}|}}{\left| \bm{n}_{t, \ell} - \left( \bm{n}_{t, \ell} \cdot \frac{\bm{p}}{|\bm{p}|} \right)  \frac{\bm{p}}{|\bm{p}|} \right|},
\notag
\end{eqnarray}
where $\mathcal{E}_{\ell,\bm{p}}$ denotes the complex amplitude of the induced electric field, incorporating both magnitude and phase; $\bm{n}_{\bm{\mathcal{E}}_{\ell,\bm{p}}}$ denotes the polarization direction of the electromagnetic wave at point $\bm{p}$; $c$ is the speed of light; $\mu_1$ is the magnetic permeability of air, $\lambda$ is the wavelength.
\end{thm}

\begin{table*}[!t]
\centering
\caption{A summary of the definitions and their geometric relationships.}
\begin{tabular}{cc|cc}
     \toprule
     \multicolumn{2}{c|}{\textbf{For $\ell$-th transmitting antenna}} & \multicolumn{2}{c}{\textbf{For $\kappa$-th receiving antenna}} \\
     \midrule
     \textbf{Position} & \makecell[l]{$\bm{p}_{t, \ell} = \left( x_{t, \ell}, y_{t, \ell}, z_{t, \ell}, \right)$} &
     \textbf{Position} & \makecell[l]{$\bm{p}_{r, \kappa} = \left( x_{r, \kappa}, y_{r, \kappa}, z_{r, \kappa} \right)$} \\
     \textbf{Direction} & \makecell[l]{$\bm{n}_{t, \ell} = \left( x_{t, \ell}^{(n)}, y_{t, \ell}^{(n)}, z_{t, \ell}^{(n)}, \right)$\\
     Polar Angle: $\theta_{t, \ell}^{(n)}$, Azimuthal Angle: $\phi_{t, \ell}^{(n)}$}&
     \textbf{Direction} & \makecell[l]{$\bm{n}_{r, \kappa} = \left( x_{r, \kappa}^{(n)}, y_{r, \kappa}^{(n)}, z_{r, \kappa}^{(n)}, \right)$\\
     Polar Angle: $\theta_{r, \kappa}^{(n)}$, Azimuthal Angle: $\phi_{r, \kappa}^{(n)}$} \\
     \multirow{2}{*}{\textbf{Emission Angle}} &\makecell[l]{\multirow{2}{*}{$\theta_{\ell, \kappa}^{(e)}$, with $\cos{\theta_{\ell, \kappa}^{(e)}} = \frac{\bm{p}_{r, \kappa} \cdot \bm{n}_{t, \ell}}{|\bm{p}_{r, \kappa}| |\bm{n}_{t, \ell}|}$}}
     & \textbf{Incident Angle} & \makecell[l]{$\theta_{\kappa}^{(i)}$, with $\sin{\theta_{\kappa}^{(i)}} = \frac{\bm{p}_{r, \kappa} \cdot \bm{n}_{r, \kappa}}{|\bm{p}_{r, \kappa}| |\bm{n}_{r, \kappa}|}$} \\
      & &  \textbf{Transmitted Angle} & \makecell[l]{$\theta_{\kappa}^{(t)}$, with $\frac{\sin{\theta_{\kappa}^{(t)}}}{\sin{\theta_{\kappa}^{(i)}}} = \sqrt{\frac{\varepsilon_1 \mu_1}{\varepsilon_2 \mu_2}}$}\\
     \bottomrule
\end{tabular}
\label{tab:notations}
\end{table*}

\begin{proof}
We start by considering a reference scenario where the dipole antenna is stationary, located at the origin $\bm{p}_{t,\ell}=\bm{o}$, and oriented along the direction $\bm{n}_{t,\ell}=\bm{k}$.
When transmitting signal $s_\ell$, the current density $\bm{J}$ applied to the dipole can be written as \cite{stutzman2012antenna}
\begin{eqnarray} \label{e:J}
\bm{J} = I_0 s_\ell \sin{\left( \frac{2\pi}{\lambda} \left( \frac{\lambda}{4} - |z| \right) \right)} \bm{k}, \quad |z| \leq \frac{\lambda}{4}.
\end{eqnarray}
where $I_0$ is the maximum current. Without loss of generality, we set $I_0=1$.

In the mmWave band, signal diffraction and scattering are minimal, making non-line-of-sight (NLOS) paths significantly weaker than the LOS path.\footnote{As a concrete example, LOS signals can be over $50$ dB stronger than the NLOS signal at a distance of 100 meters, according to the mmWave channel specified in 3GPP TR 38.901 \cite{3GPP5G}.} Therefore, we only consider the LOS component. 

The electric field radiated by the stationary transmitting antenna at any point $\bm{p} = (x_p, y_p, z_p)$ in space is given by 
\begin{equation} \label{e:E0}
\bm{\mathcal{E}}_{\bm{o},\bm{p}} = 
2 j c \mu_1 s_\ell \frac{e^{-j \frac{2 \pi}{\lambda} |\bm{p}|} }{4 \pi |\bm{p}| } \frac{\cos{\left( \frac{\pi}{2} \cos{\theta_{\bm{o},\bm{p}}^{(e)}} \right)}}{\sin{\theta_{\bm{o},\bm{p}}^{(e)}}} \bm{\vartheta},
\end{equation}
where $\theta_{\bm{o},\bm{p}}^{(e)}$ is the emission angle between $\bm{k}$ and $\bm{p}$, and  $\cos{\theta_{\bm{o},\bm{p}}^{(e)}} = \frac{\bm{p} \cdot \bm{k}}{|\bm{p}||\bm{k}|} = \frac{z_p}{|\bm{p}|}$. The unit vector $\bm{\vartheta}$ is the polarization direction at $\bm{p}$.
Eq. \eqref{e:E0} is a classical result in electromagnetic theory \cite{stutzman2012antenna,stratton2007electromagnetic} and provides the foundation for our analysis. In Appendix \ref{sec:AppA} we provide a detailed derivation of this equation from Maxwell's equations.

We next consider the antenna movement. After translation and rotation, the $\ell$-th antenna moves to  $\bm{p}_{t,\ell}$ with a direction $\bm{n}_{t,\ell}$. The complex amplitude of the induced electric field follows from \eqref{e:E0} and can be written as
\begin{equation} \label{e:El_amp}
\mathcal{E}_{\ell,\bm{p}} = 2 j c \mu_1 s_\ell \frac{e^{-j \frac{2 \pi}{\lambda} \left|\bm{p} - \bm{p}_{t,\ell}\right|} }{4 \pi \left|\bm{p} - \bm{p}_{t,\ell}\right| } \frac{\cos{\left( \frac{\pi}{2} \cos{\theta^{(e)}_{\ell,\bm{p}}} \right)}}{\sin{\theta^{(e)}_{\ell,\bm{p}}}},
\end{equation}
where $\theta^{(e)}_{\ell,\bm{p}}$ is the emission angle from $\bm{p}_{t,\ell}$ to $\bm{p}$.
Under the far-field condition, $\bm{p}$ and $\bm{p} - \bm{p}_{t,\ell}$ can be treated parallel, hence we have
\begin{eqnarray*}
\left|\bm{p} - \bm{p}_{t,\ell}\right| \hspace{-0.2cm}&\approx&\hspace{-0.2cm}
|\bm{p}| - \left|\bm{p}_{t,\ell}\right| \cos{\left\langle  \bm{p}, \bm{p}_{t,\ell} \right\rangle} = |\bm{p}| - \frac{1}{|\bm{p}|} \bm{p} \cdot \bm{p}_{t,\ell},
\\
\cos{\theta^{(e)}_{\ell,\bm{p}}} 
\hspace{-0.2cm}&=&\hspace{-0.2cm}
\frac{\left(\bm{p} - \bm{p}_{t, \ell} \right) \cdot \bm{n}_{t,\ell}}{\left|\left(\bm{p} - \bm{p}_{t, \ell} \right)\right||\bm{n}_{t,\ell}|} \approx \frac{1}{|\bm{p}|} \bm{p} \cdot \bm{n}_{t,\ell}.
\end{eqnarray*}

Therefore, \eqref{e:El_amp} can be refined as
\begin{equation} \label{e:El_amp1}
\mathcal{E}_{\ell,\bm{p}} \approx 2 j c \mu_1 s_\ell \frac{e^{-j \frac{2 \pi}{\lambda} |\bm{p}|} }{4 \pi |\bm{p}| } \frac{\cos{\left( \frac{\pi}{2} \cos{\theta^{(e)}_{\ell,\bm{p}}} \right)}}{\sin{\theta^{(e)}_{\ell,\bm{p}}}}e^{j \frac{2\pi}{\lambda} \frac{\bm{p} \cdot \bm{p}_{t,\ell}}{|\bm{p}|}},
\end{equation}
where the approximation follows from $\lim \limits_{|\bm{p}_{t,\ell}| / |\bm{p}| \rightarrow 0}\big(|\bm{p}| - \frac{\bm{p} \cdot \bm{p}_{t,\ell}}{|\bm{p}|} \big) = |\bm{p}|$.
Note that the term $\frac{2\pi}{\lambda}$ in the exponent indicates that the phase varies with distance on the order of the wavelength, hence the term $\frac{\bm{p} \cdot \bm{p}_{t,\ell}}{|\bm{p}|}$ in the exponent cannot be neglected.

At point $\bm{p}$, the polarization direction of the electric field is perpendicular to the propagation direction, and lies in the plane formed by the direction of the transmitting antenna $\bm{n}_{t, \ell}$ and the propagation direction.
Thus, the polarization direction $\bm{n}_{\bm{\mathcal{E}}_{\ell,\bm{p}}}$ can be obtained by subtracting the component along the propagation direction from the direction of the transmitting antenna, yielding
\begin{equation} \label{e:El_dir}
\bm{n}_{\bm{\mathcal{E}}_{\ell,\bm{p}}} = \frac{\bm{n}_{t, \ell} - \left( \bm{n}_{t, \ell} \cdot \frac{\bm{p}}{|\bm{p}|} \right)  \frac{\bm{p}}{|\bm{p}|}}{\left| \bm{n}_{t, \ell} - \left( \bm{n}_{t, \ell} \cdot \frac{\bm{p}}{|\bm{p}|} \right)  \frac{\bm{p}}{|\bm{p}|} \right|}.
\end{equation}

Combining \eqref{e:El_amp} and \eqref{e:El_dir} gives us the induced electric field at observation point $\bm{p}$ in \eqref{e:El}.
\end{proof}

Eq. \eqref{e:El} reveals that, at a given observation point $\bm{p}$, both the complex amplitude and the polarization of the electric field are affected by the movement of the transmitting antenna. 


\begin{cor}\label{cor:1}
The magnitude of the electric field
\begin{equation} \label{e:magnitude}
|\mathcal{E}_{\ell,\bm{p}}| \triangleq \frac{ 2 c \mu_1 s_\ell }{4 \pi |\bm{p}| } \frac{\cos{\left( \frac{\pi}{2} \cos{\theta^{(e)}_{\ell,\bm{p}}} \right)}}{\sin{\theta^{(e)}_{\ell,\bm{p}}}} 
\end{equation}
first increases and then decreases as a function of $\theta^{(e)}_{\ell,\bm{p}}$, reaching its maximum value when $\theta^{(e)}_{\ell,\bm{p}} = \frac{\pi}{2}$. This variation depends solely on the emission angle, which results from the 3D orientation of the transmitting antenna.
\end{cor}

\begin{proof}
    See Appendix~\ref{sec:AppB}.
\end{proof}

This paper focuses on the half-wave dipole, the most fundamental type of antenna, to develop the PAMA framework. The dipole antenna is directional in the 3D space: when positioned vertically, it radiates uniformly in the horizontal plane but has directional properties in the vertical plane. Our PAMA framework is sufficiently versatile to extend to more complex directional antennas, provided their radiation patterns are known.

\subsection{Electric Field Sensing}
\label{Sec:SM_Rx}
Theorem \ref{thm:radiation} establishes the electric field induced by a movable transmitting antenna. In this section, we explore how a movable receiving antenna senses the electric field. 

The process of capturing radio waves with a receiving antenna involves the transmission of electromagnetic waves from air into the antenna. Dipoles emit linearly polarized radio waves, and when such a plane wave reaches the boundary between the air and the antenna, part of the wave is reflected back into the air, while the remainder is transmitted into the antenna.

\begin{defi}[Incident Angle]\label{defi:Incident}
The incident angle $\theta^{(i)} \in [0, \frac{\pi}{2}]$ is the angle between the direction of the incoming electromagnetic wave and the normal to the surface of the antenna at the point of impact. The direction of the normal $\bm{n}_{r, \kappa}^{\perp}$ can be obtained by subtracting the component that is collinear with the receiving antenna $\bm{n}_{r,\kappa}$ from the direction of the incident electromagnetic wave $\bm{p}_{r,\kappa}$:
\begin{equation}
  \bm{n}_{r, \kappa}^{\perp} = \frac{\bm{p}_{r, \kappa} - ( \bm{p}_{r, \kappa} \cdot \bm{n}_{r, \kappa} ) \bm{n}_{r, \kappa}}{\left| \bm{p}_{r, \kappa} - \left( \bm{p}_{r, \kappa} \cdot \bm{n}_{r, \kappa} \right) \bm{n}_{r, \kappa} \right|}.  
\end{equation}
We denote by $\theta_{\ell,\kappa}^{(i)}$ the incident angle of the wave from the $\ell$-th transmitting antenna to the $\kappa$-th receiving antenna.
\end{defi}

\begin{defi}[Transmitted Angle]
The transmitted angle $\theta^{(t)} \in [0, \frac{\pi}{2}]$ is the angle between the normal of the receiving antenna and the new propagation direction of the electromagnetic wave inside the antenna after it has entered. 
\end{defi}

\begin{figure}[!t]
    \centering
    \includegraphics[width=.68\linewidth]{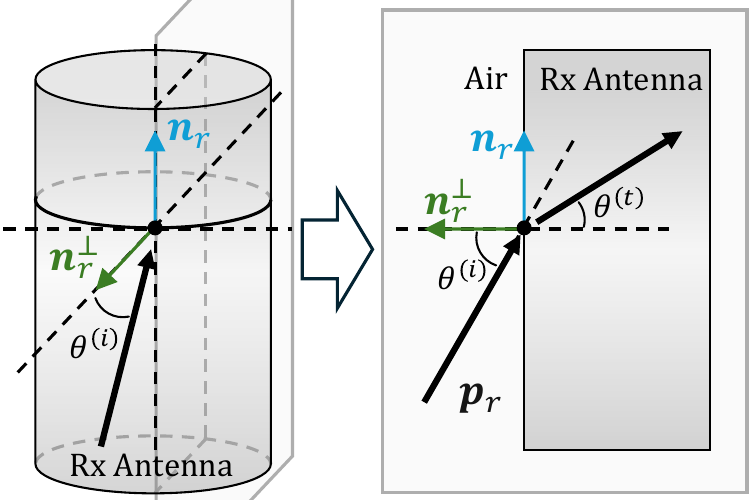}
    \caption{Illustration of the incident angle $\theta^{(i)}$ and transmitted angle $\theta^{(t)}$.}
    \label{fig:transmitted_angle}
\end{figure}

An illustration of the incident angle and transmitted angle is shown in Fig.~\ref{fig:transmitted_angle}

\begin{defi}[Permittivity and Permeability]
Permittivity $\varepsilon$ quantifies the ability of a material to store electrical energy in an electric field. Permeability $\mu$ is a measure of how easy or difficult it is for a magnetic field to pass through a material \cite{stratton2007electromagnetic}. We define $\varepsilon_1$ and $\varepsilon_2$ as the permittivity of air and antenna, respectively, and $\mu_1$ and $\mu_2$ as the permeability of air and antenna, respectively. We further define a relative
permittivity 
\begin{equation}
    \varepsilon_r \triangleq \frac{\varepsilon_2}{\varepsilon_1}>1.
\end{equation}
\end{defi}

\begin{lem}[Snell’s law \cite{parazzoli2003experimental}]
Let $\varepsilon_1$ and $\varepsilon_2$ be the permittivity of air and antenna, respectively, and $\mu_1$ and $\mu_2$ be the permeability of air and antenna, respectively. We have
\begin{equation}
    \frac{\sin{\theta^{(t)}}}{\sin{\theta^{(i)}}} = \sqrt{\frac{\varepsilon_1 \mu_1}{\varepsilon_2 \mu_2}}.
\end{equation}    
\end{lem}

A critical aspect of electric field sensing is polarization matching, which refers to the alignment of the transmitting and receiving antennas' polarization states. The efficiency of polarization matching determines how much energy is transferred into the antenna. Therefore, effective polarization matching is essential for maximizing the power transfer and minimizing signal degradation, especially in systems with movable antennas where orientation can dynamically change.

\begin{defi}[Polarization Matching Angle]
The polarization matching angle $\alpha \in [0, \pi]$ is the angle between the polarization direction of the incident wave and the direction of the receiving antenna. We denote by  $\alpha_{\ell,\kappa}=\langle \bm{n}_{\bm{\mathcal{E}}_{\ell,\bm{p}_{r,\kappa}}},\bm{n}_{r, \kappa}\rangle$ the polarization matching angle of the wave emitted from the $\ell$-th transmitting antenna to the $\kappa$-th receiving antenna.
\end{defi}

Theorem \ref{thm:sensing} below summarizes the main result of this section.

\begin{thm}[Electric Field Sensing]
\label{thm:sensing}
Consider the signal transmitted from the $\ell$-th transmitting antenna. The received signal at the $\kappa$-th user' antenna can be written as
\begin{eqnarray} \label{e:rt}
&&\hspace{-0.6cm} r_{\ell, \kappa} =  \frac{\mathcal{E}_{\ell, \bm{p}_{r,\kappa}}}{A_F}\mathcal{M}\left(\theta_{\ell, \kappa}^{(i)}, \alpha_{\ell, \kappa}\right) , \\
&&\hspace{-0.6cm} \mathcal{M}\left(\theta_{\ell, \kappa}^{(i)}, \alpha_{\ell, \kappa}\right)
\!=\! \sqrt{ 1 \!-\! \left| \Gamma_{\parallel, \kappa} \right|^2 \cos^2{\alpha_{\ell, \kappa}} \!-\! \left| \Gamma_{\perp, \kappa} \right|^2  \sin^2{\alpha_{\ell, \kappa}}}, \notag
\end{eqnarray}
where $\mathcal{E}_{\ell,\bm{p}_{r,\kappa}}$ is the complex amplitude of the induced
electric field by the $\ell$-th antenna at $\bm{p}_{r,\kappa}$; 
the constant $A_F$ is an antenna factor for converting field amplitude to voltage; 
$\mathcal{M}\left(\theta_{\ell, \kappa}^{(i)}, \alpha_{\ell, \kappa}\right)$ denotes the polarization matching efficiency, considering both the polarization matching angle and energy reflection; 
$\Gamma_{\parallel, \kappa}$ and $\Gamma_{\perp, \kappa}$ are the energy reflection coefficients for the horizontal and vertical polarization:
\begin{eqnarray} \label{e:Gamma_para}
&&\hspace{-1cm} \Gamma_{\parallel, \kappa} = \frac{ \sqrt{\varepsilon_r - 1 + \cos^2{\theta_{\ell, \kappa}^{(i)}}} - \varepsilon_r \cos{\theta_{\ell, \kappa}^{(i)}}}{ \sqrt{\varepsilon_r - 1 + \cos^2{\theta_{\ell, \kappa}^{(i)}}} + \varepsilon_r \cos{\theta_{\ell, \kappa}^{(i)}}}, \\
\label{e:Gamma_perp}
&&\hspace{-1cm}\Gamma_{\perp, \kappa} = \frac{ \sqrt{\varepsilon_r - 1 + \cos^2{\theta_{\ell, \kappa}^{(i)}}} - \cos{\theta_{\ell, \kappa}^{(i)}}}{ \sqrt{\varepsilon_r - 1 + \cos^2{\theta_{\ell, \kappa}^{(i)}}} + \cos{\theta_{\ell, \kappa}^{(i)}}}.
\end{eqnarray}
\end{thm}

\begin{proof}
The electric field induced by the $\ell$-th transmitting antenna possesses a polarization direction $\bm{n}_{\bm{\mathcal{E}}_{\ell,\bm{p}_{r,\kappa}}}$. The $\kappa$-th receiving antenna, oriented at $\bm{n}_{r, \kappa}$, has a polarization matching angle $\alpha_{\ell,\kappa}$ w.r.t. the incident wave. We can partition the incident polarization into horizontal and vertical components.

Due to the electromagnetic boundary conditions at the interface of different media (air to antenna), part of the energy of each polarization component is reflected. The reflection coefficients $\Gamma_{\parallel, \kappa}$ and $\Gamma_{\perp, \kappa}$, corresponding to horizontal and vertical polarizations respectively, determine how much of the incident wave's energy is transmitted into the receiver.

The reflection coefficient $\Gamma$ for any polarization can generally be expressed in terms of the characteristic impedances of the interacting media \cite{stratton2007electromagnetic,huang2021antennas}:
\begin{equation} \label{e:Gamma}
\Gamma = \frac{Z_2 - Z_1}{Z_2 + Z_1},
\end{equation}
where 
\begin{eqnarray*} \label{e:impedance}
&&\hspace{-1cm} Z_1 = \begin{cases}
120 \pi \sqrt{\frac{\mu_1}{\varepsilon_1}} \cos{\theta^{(i)}_{\ell,\kappa}}, & \text{for horizontal polarization}, \\
120 \pi \sqrt{\frac{\mu_1}{\varepsilon_1}} / \cos{\theta^{(i)}_{\ell,\kappa}}, & \text{for vertical polarization},
\end{cases} \notag\\
&&\hspace{-1cm} Z_2 = \begin{cases}
120 \pi \sqrt{\frac{\mu_2}{\varepsilon_2}} \cos{\theta^{(t)}_{\ell,\kappa}}, & \text{for horizontal polarization}, \\
120 \pi \sqrt{\frac{\mu_2}{\varepsilon_2}} / \cos{\theta^{(t)}_{\ell,\kappa}}, & \text{for vertical polarization}.
\end{cases}
\end{eqnarray*}

Given \eqref{e:Gamma}, we compute the reflection coefficients for the horizontal and vertical components:
\begin{eqnarray} \label{e:Gamma_para0}
&&\hspace{-1cm} \Gamma_{\parallel} = \frac{ 120 \pi \sqrt{\frac{\mu_2}{\varepsilon_2}} \sqrt{1 - \sin^2{\theta^{(t)}_{\ell,\kappa}}} -  120 \pi \sqrt{\frac{\mu_1}{\varepsilon_1}} \cos{\theta^{(i)}_{\ell,\kappa}}}{ 120 \pi \sqrt{\frac{\mu_2}{\varepsilon_2}} \sqrt{1 - \sin^2{\theta^{(t)}_{\ell,\kappa}}} + 120 \pi \sqrt{\frac{\mu_1}{\varepsilon_1}} \cos{\theta^{(i)}_{\ell,\kappa}}} \notag\\
&&\hspace{-0.5cm} \overset{(a)}{=} \frac{ \sqrt{\varepsilon_r - \sin^2{\theta^{(i)}_{\ell,\kappa}}} - \varepsilon_r \sqrt{1 - \sin^2{\theta^{(i)}_{\ell,\kappa}}}}{ \sqrt{\varepsilon_r - \sin^2{\theta^{(i)}_{\ell,\kappa}}} + \varepsilon_r \sqrt{1 - \sin^2{\theta^{(i)}_{\ell,\kappa}}}}, \\
\label{e:Gamma_perp0}
&&\hspace{-1cm}\Gamma_{\perp} = \frac{ \sqrt{\varepsilon_r - \sin^2{\theta^{(i)}_{\ell,\kappa}}} - \sqrt{1 - \sin^2{\theta^{(i)}_{\ell,\kappa}}}}{ \sqrt{\varepsilon_r - \sin^2{\theta^{(i)}_{\ell,\kappa}}} + \sqrt{1 - \sin^2{\theta^{(i)}_{\ell,\kappa}}}},
\end{eqnarray}
where (a) follows because the permeability of most materials remains constant and is unaffected by changes in frequency or temperature \cite{huang2021antennas}. Therefore, we have $\mu_1=\mu_2$.

\begin{figure*}[!t]
    \centering
    \includegraphics[width=0.88\linewidth]{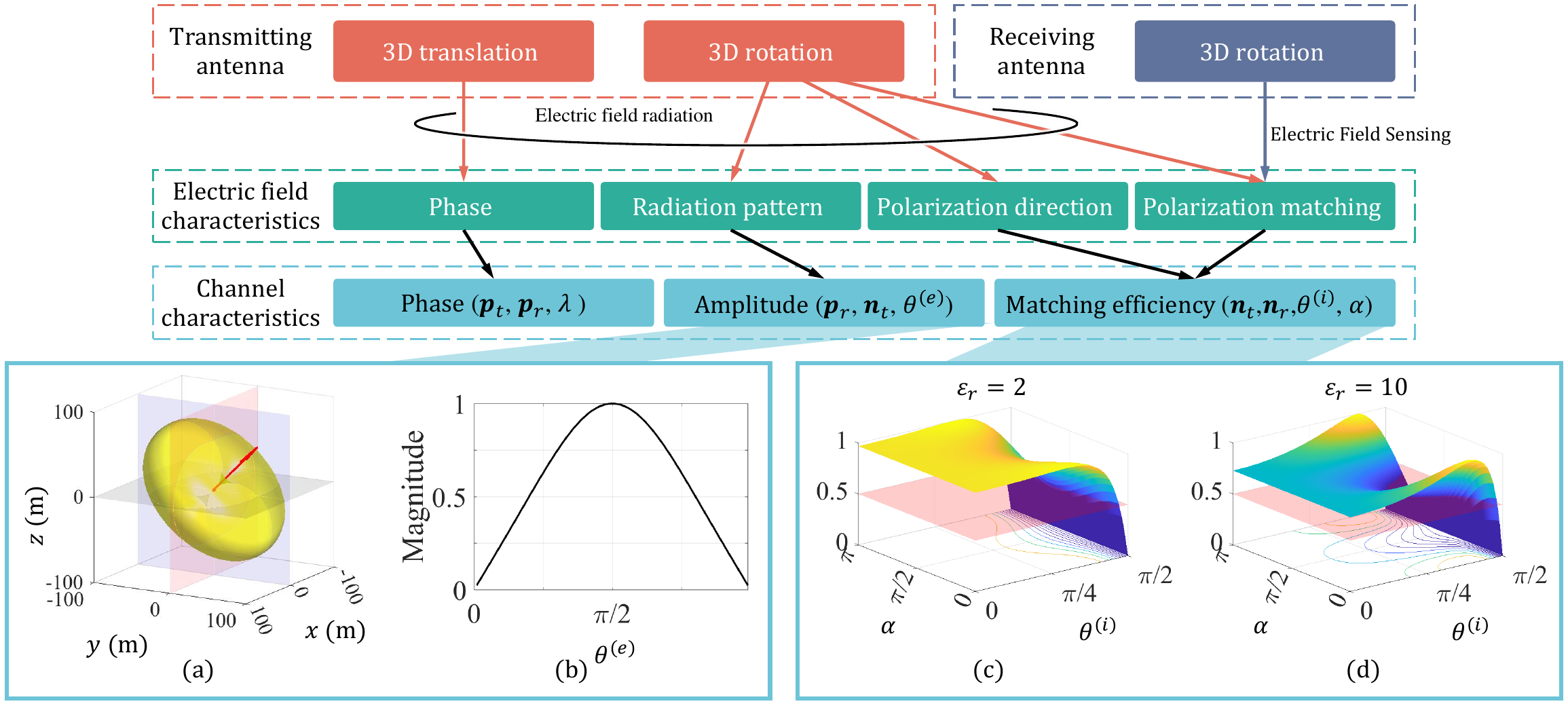}
    \caption{Relationships between the PAMA system, the electric field, and the PAMA channel gain under the far-field condition with a single LOS path. (a) the impact of rotating the transmitting antenna on the radiation pattern relative to the CCS; (b) variation in received signal strength as a function of the emission angle; (c), (d): changes in polarization matching efficiency based on the incident angle $\theta^{(i)}$ and the polarization matching angle $\alpha$.}
    \label{f:relationship}
\end{figure*}

Without loss of generality, this paper represents the received signal by the voltage induced at the receiver. In practice, the difference between the amplitude of the received signal in terms of voltage and current is simply a scaling factor: the resistance, which can be considered part of the antenna factor. Therefore, the received signal, i.e., the induced voltage, at the $\kappa$-th user from the $\ell$-th transmitting antenna can be written as
\begin{eqnarray*}
r_{\ell, \kappa} = V_0 =  \frac{\mathcal{E}_{\ell,\bm{p}_{r,\kappa}}}{A_F}\mathcal{M}\left(\theta_{\ell, \kappa}^{(i)}, \alpha_{\ell, \kappa}\right),
\end{eqnarray*}
where the constant $A_F$ is the antenna factor for converting field amplitude to voltage, as shown in \eqref{e:rt}.
\end{proof}

\begin{rem}[Comparison with 3GPP TR 38.901]
The polarization matching formulation in the PAMA model generalizes and refines the modeling approach adopted in 3GPP TR 38.901 \cite{3GPP5G}. In the 3GPP model, the polarization matching effect is represented by the inner product between the transmit and receive antenna field patterns, which assumes that the aligned polarization component is fully absorbed, while orthogonal components are completely rejected. This binary treatment fails to capture the fact that each polarization state can lead to partial energy absorption and reflection. In contrast, the PAMA model incorporates a physically grounded characterization of electromagnetic wave interaction at the antenna boundary. By introducing polarization-dependent reflection coefficients, i.e., $\Gamma_{\parallel, \kappa}$ and $\Gamma_{\perp, \kappa}$, derived from Maxwell's boundary conditions, the PAMA model enables a continuous and more realistic quantification of received energy as a function of both polarization alignment and incident angle. This distinction is particularly critical in systems with movable or misaligned antennas, where even partially aligned polarization can contribute to the received signal. As such, the PAMA model provides a more accurate and generalizable representation of polarization matching efficiency in dynamic wireless environments.
\end{rem}


\subsection{Polarization-aware channel gain}
Building upon Theorems \ref{thm:radiation} and \ref{thm:sensing}, this section characterizes the polarization-aware channel gain for movable transmitting and receiving antennas.

\begin{thm}[PAMA Channel Gain]\label{thm:channel_gain}
Consider the transmission from the $\ell$-th transmitting antenna to the $\kappa$-th receiving antenna. The polarization-aware channel gain is given by
\begin{eqnarray} \label{e:h}
&&\hspace{-1cm} h_{\kappa, \ell}\left(\bm{p}_{t, \ell}, \bm{n}_{t, \ell}, \bm{p}_{r, \kappa}, \bm{n}_{r, \kappa}\right) = \frac{r_{\ell, \kappa}}{s_\ell} \\
&&\hspace{-0.5cm} = G_{\kappa} \frac{\cos{\left( \frac{\pi}{2} \cos{\theta^{(e)}_{\ell, \kappa}} \right)}}{\sin{\theta^{(e)}_{\ell, \kappa}}} \mathcal{M}\left(\theta_{k}^{(i)}, \alpha_{\ell, \kappa}\right) e^{j \frac{2\pi}{\lambda}  \frac{\bm{p}_{r, \kappa} \cdot \bm{p}_{t,\ell}}{|\bm{p}_{r,\kappa}|} }, \notag
\end{eqnarray}
where $ G_{\kappa} \triangleq \frac{2 j c \mu_1}{A_F} \frac{e^{-j \frac{2\pi}{\lambda} |\bm{p}_{r,\kappa}|}}{4 \pi |\bm{p}_{r,\kappa}|}$.
\end{thm}

Theorem \ref{thm:channel_gain} consolidates the core elements of the PAMA framework into a unified expression for the channel gain, explicitly capturing how antenna movements, including both translation and rotation, affect the received signal through phase variation, directional radiation patterns, and polarization matching efficiency.  In contrast to traditional fixed-antenna channel models, the PAMA model offers a physically grounded, continuous characterization of polarization interaction, derived from Maxwell's equations. This enables a more precise understanding of how antenna orientation and incidence geometry influence signal strength, particularly in dynamic or high-frequency scenarios. In the next section, we analyze how each component of the PAMA channel gain, namely phase, emission angle, and matching efficiency, responds to antenna movement, and how these effects collectively govern the system's communication performance.

%% file: SecIV.tex
This section analyzes the PAMA channel model and examines how each variable influences channel quality, providing insights into the key characteristics of PAMA and identifying the sources of its performance gains. Fig.~\ref{f:relationship} summarizes the relationships between antenna movement and their corresponding impacts on the channel.
It is essential to highlight that these properties are derived under the conditions of a far-field scenario and a dominant LOS path.
For clarity, the indices $\ell$ and $\kappa$ in the subscripts are omitted in this section.

\begin{prop}\label{prop:1}
In a single-antenna system, antenna translation only contributes a phase term to the PAMA channel gain, implying that its impact can be entirely captured through precoding.
\end{prop}

In the far-field LOS scenario, the electromagnetic waves can be approximated as plane waves, and the spatial fading characteristics are negligible. The relative orientation of the electric field vectors between the transmitter and receiver remains unchanged despite antenna translation. This means that antenna translation does not affect the amplitude or polarization alignment of the received signal -- only the phase changes due to the difference in path length, leading to Proposition \ref{prop:1}.

In practical systems, digital precoding is often easier to implement than mechanical antenna movements. Therefore, when studying MAs, it is essential to consider the unique advantages that MAs introduce beyond what precoding can achieve. Proposition \ref{prop:1} indicates that the benefits of antenna translation are more significant in near-field or rich multi-path environments, in which case antenna translation can impact the amplitude and spatial characteristics of the channel gain, offering advantages that precoding alone cannot replicate.

In the context of this paper, which focuses on far-field and mmWave dominant LOS paths, the unique gains that MAs can offer over precoding are primarily obtained through antenna rotation rather than translation.

\begin{prop}\label{prop:2}
The rotation of the transmitting antenna alters the amplitude of the channel gain by modifying both the radiation pattern of the induced electric field and the polarization matching efficiency at the receiver. 
\end{prop}

When the transmitting antenna rotates, it changes the emission angle $\theta^{(e)}$ of the electromagnetic wave. This alteration affects the antenna's radiation pattern -- the distribution of the electric field w.r.t. the CCS -- which in turn influences the electric field intensity received. As illustrated in Fig.~\ref{f:relationship} (a) and (b), different emission angles lead to variations in the radiation pattern, resulting in changes to the received signal's strength. According to Corollary~\ref{cor:1}, the magnitude or intensity of the electric field reaches its maximum when the emission angle is $\theta^{(e)} = \frac{\pi}{2}$. To achieve this, the transmitting antenna should be positioned within the plane that is perpendicular to the direction of wave propagation.

Furthermore, the rotation of the transmitting antenna also modifies its polarization orientation, thereby changing the polarization matching angle $\alpha$ between the transmitter and receiver, and hence the matching efficiency $\mathcal{M}\left(\theta ^{(i)}, \alpha \right)$ in the PAMA channel gain.

\begin{prop}\label{prop:3}
The rotation of the receiving antennas alters the amplitude of the channel gain by modifying the polarization matching efficiency.
\end{prop}

Rotating the receiving antennas affects both the incident angle $\theta^{(i)}$ of the incoming electromagnetic wave and the polarization matching angle $\alpha$, both of which impacts the polarization matching efficiency.
\begin{enumerate}
    \item The change of the incident angle $\theta^{(i)}$ alters the energy reflection coefficients $\Gamma_{\parallel}$ and $\Gamma_{\perp}$, as defined in \eqref{e:Gamma_para} and \eqref{e:Gamma_perp}. These reflection coefficients characterize the proportion of electromagnetic energy that enters the antenna. A smaller reflection coefficient indicates that more energy is transmitted into the antenna, enhancing the received signal strength.
    \item The orientation of the receiving antenna relative to the electric field vector of the incoming wave determines the polarization matching angle $\alpha$, which describes the ratio of the horizontal to vertical components of the electromagnetic wave's polarization relative to the antenna's orientation. A better alignment (i.e., a smaller mismatch angle) leads to higher polarization matching efficiency.
\end{enumerate}

Figs.~\ref{f:relationship} (c) and (d) illustrate how the polarization matching efficiency varies with the incident angle $\theta^{(i)}$ and the polarization matching angle $\alpha$ for antennas with different relative dielectric constants $\varepsilon_r$. As can be seen, without proper optimization of the receiving antenna's rotation, the matching efficiency can decrease significantly, potentially dropping to zero. In such cases, the antenna fails to effectively receive the incoming signal, resulting in a loss of information. Therefore, careful control of the receiving antenna's orientation is essential for maximizing the channel gain and ensuring reliable communication in the PAMA framework.

\begin{figure}[!t]
\centering
    \begin{subfigure}[!t]{0.55\linewidth}
        \centering
        \includegraphics[width=1\linewidth]{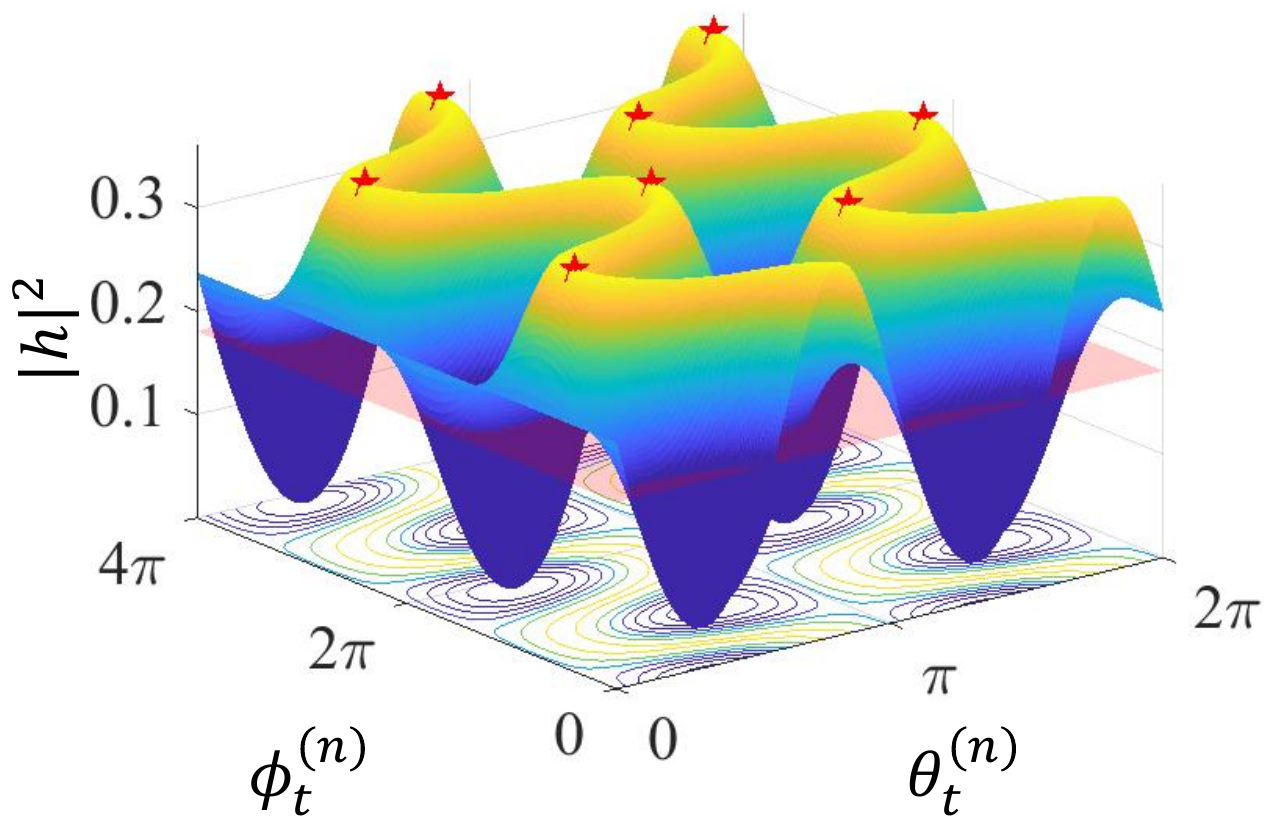}
        \caption{$|h|^2$}\label{f:P2P_nt_gain}
    \end{subfigure}
    \begin{subfigure}[!t]{0.43\linewidth}
        \includegraphics[width=1\linewidth]{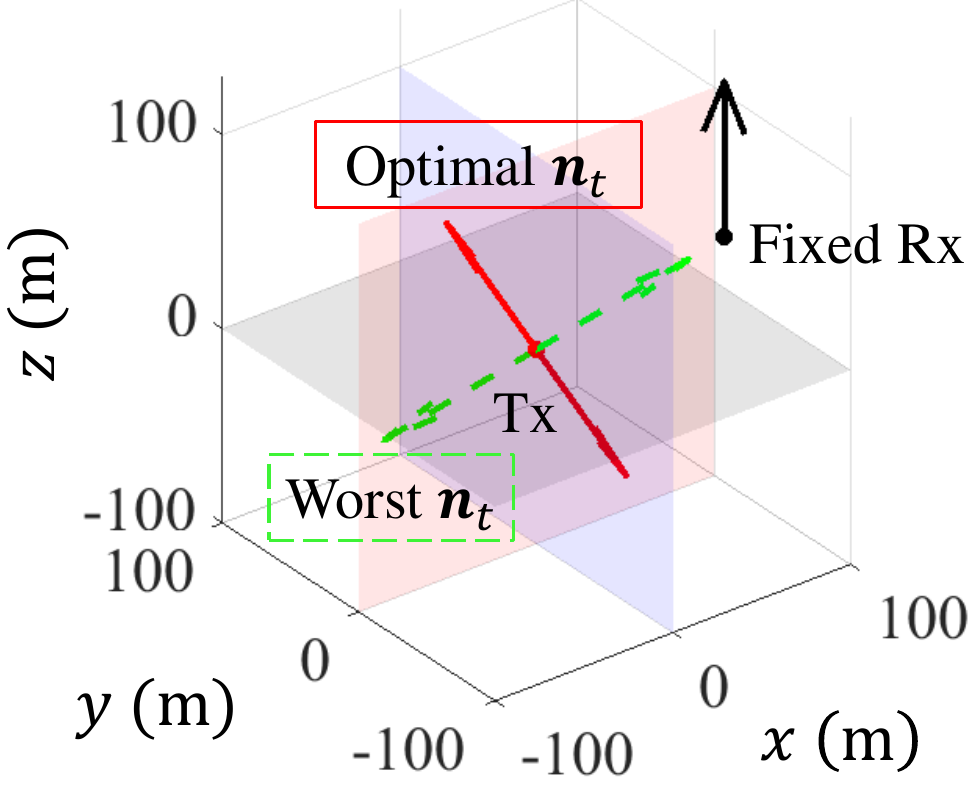}
        \caption{Geometry relationship.}\label{f:P2P_nt_geo}
    \end{subfigure}
\caption{Channel amplitude variations and the orientations of the transmitting antenna in Scenario 1, where $\bm{p}_{t,1} = [0, 0, 0]$, $\bm{p}_{r,1} = [75, -40, 50]$, and $\bm{n}_{r,1} = [0, 0, 1]$. }
\label{f:P2P_nt}
\end{figure}

\begin{figure}[!t]
\centering
    \begin{subfigure}[!t]{0.55\linewidth}
        \centering
        \includegraphics[width=1\linewidth]{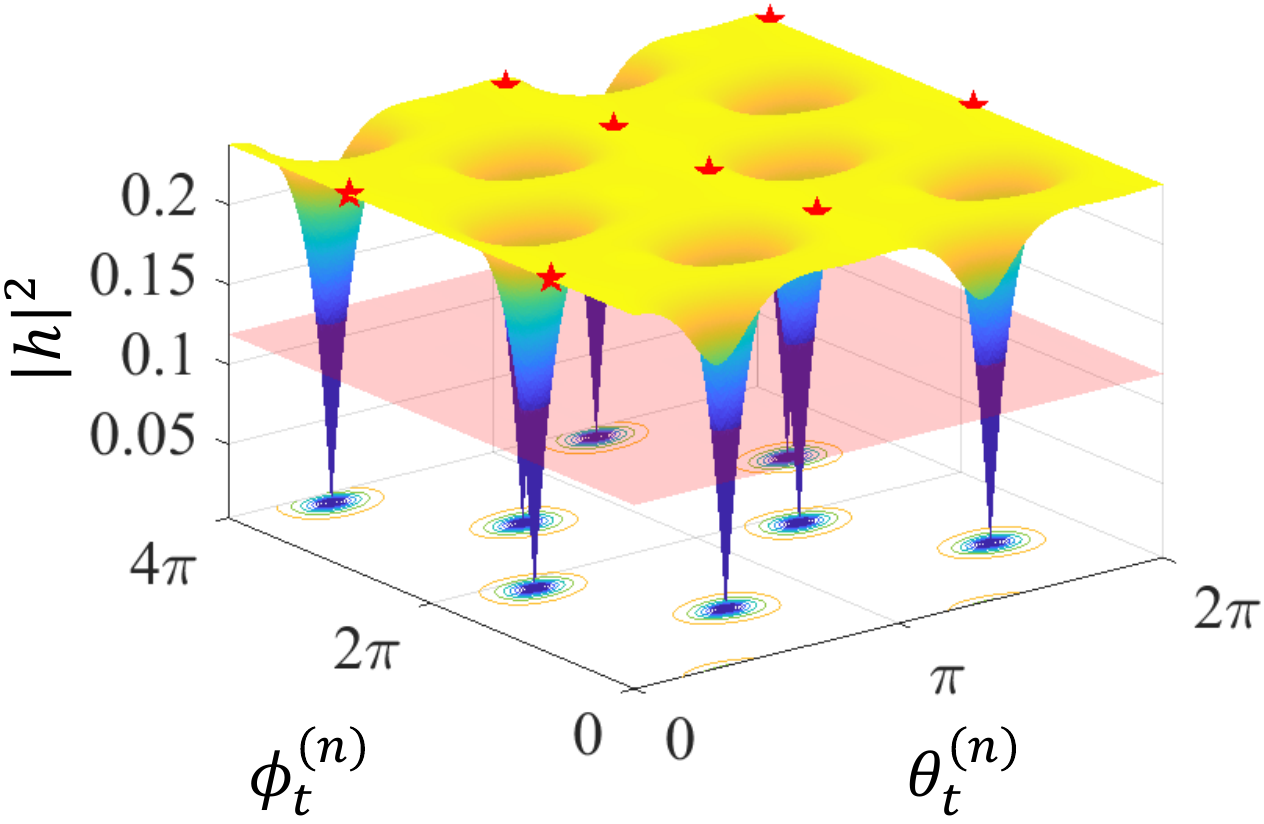}
        \caption{$|h|^2$}\label{f:P2P_nr_gain}
    \end{subfigure}
    \begin{subfigure}[!t]{0.43\linewidth}
        \includegraphics[width=1\linewidth]{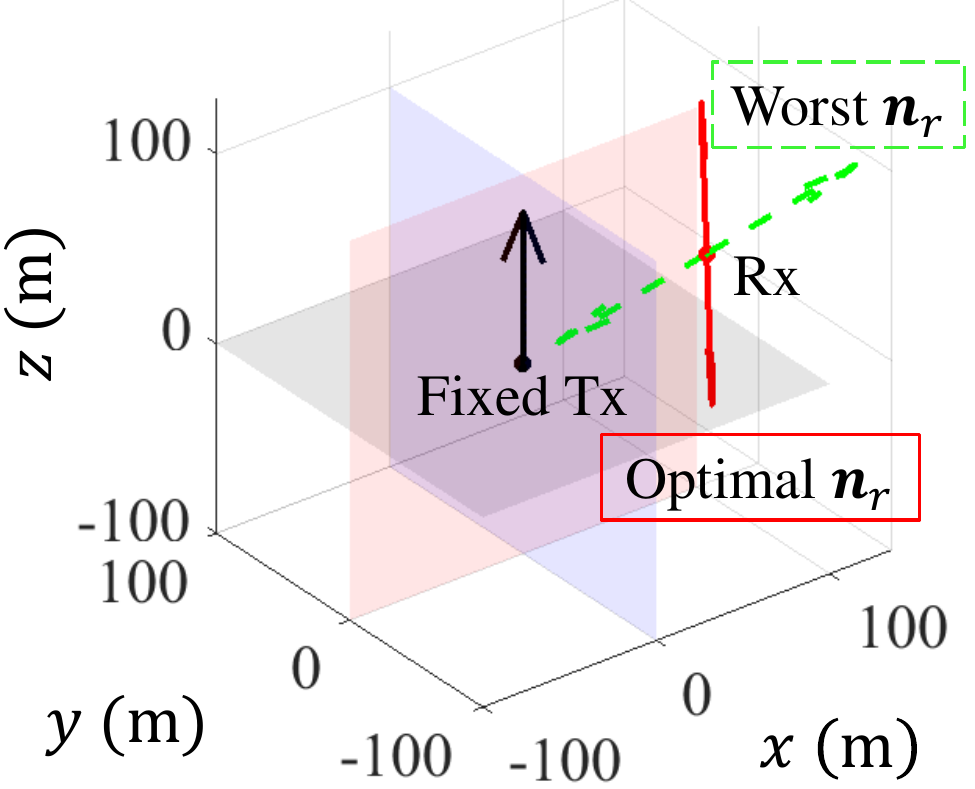}
        \caption{Geometry relationship.}\label{f:P2P_nr_geo}
    \end{subfigure}
\caption{Channel amplitude variations and the orientations of the receiving antenna in Scenario 2, where $\bm{p}_{t,1} = [0, 0, 0]$, $\bm{n}_{t,1} = [0, 0, 1]$, and $\bm{p}_{r,1} = [75, -40, 50]$.}
\label{f:P2P_nr}
\end{figure}

Propositions \ref{prop:2} and \ref{prop:3} indicate that the amplitude of the channel gain is influenced by the orientations of both the transmitting and receiving antennas. To illustrate the extent to which antenna rotations affect the channel amplitude, we simulate two scenarios:
\begin{itemize}
    \item \textbf{Scenario 1}: Rotatable transmitting antenna with fixed receiving antenna;
    \item \textbf{Scenario 2}: Fixed transmitting antenna with rotatable receiving antenna.
\end{itemize}

In the simulations, the positions of the antennas are set as $\bm{p}_{t} = [0, 0, 0]$, $\bm{p}_{r} = [75, -40, 50]$. The fixed antennas are oriented vertically upward, with a direction vector of $[0, 0, 1]$. For visualization purposes, we vary the polar and azimuthal angles of the rotatable antennas over two periods, even though their actual values lie within a single period. 
The channel amplitude variations for the two scenarios are presented in Figs.~\ref{f:P2P_nt}(a) and \ref{f:P2P_nr}(a), respectively. 

For scenario 1 with a fixed receiving antenna, Fig.~\ref{f:P2P_nt}(a) shows that a randomly oriented transmitting antenna has only a $67.5\%$ probability of delivering at least half of the maximum possible energy to the user. Each period can yield two sets of maximum values due to the symmetry of the dipole antennas, where the collinear reverse sets of polar and azimuthal angles provide equivalent performance.

For scenario 2 with a fixed transmitting antenna, Fig.~\ref{f:P2P_nr}(a) shows that a randomly oriented receiving antenna has a $99.0\%$ probability of capturing at least half of the maximum energy. However, the remaining $1.0\%$ of orientations present a significant risk of receiving virtually no energy. This outcome highlights the critical importance of antenna orientation in ensuring reliable communication.

Figs.~\ref{f:P2P_nt}(b) and \ref{f:P2P_nr}(b) depict the positions and orientations of the transmitting and receiving antennas for both scenarios. In these illustrations, fixed Antennas are shown with vertical black arrows; the optimal orientations of the rotatable antennas are shown with solid red arrows, indicating the directions that maximize the channel gain; the worst-case orientations of the rotatable antennas are shown with dashed green arrows, representing the directions that minimize the channel gain.

\begin{prop} \label{prop:4}
In a single-input single-output (SISO) PAMA system:
\begin{enumerate}
    \item Optimal orientation: The best channel gain is achieved when the transmitting and receiving antennas are oriented such that they, along with the propagation path, lie within the same plane.
    \item Worst orientation: The poorest channel gain occurs when one antenna points directly at the other.
\end{enumerate}
\end{prop}

When the electromagnetic field reaches the receiving antenna, it can be decomposed into three orthogonal components along: the direction of the receiving antenna $\bm{n}_r$; the normal direction of the incident signal $\bm{n}_{r}^{\perp}$; a direction perpendicular to both $\bm{n}_r$ and $\bm{n}_{r}^{\perp}$, denoted by $\bm{n}_3$.
This decomposition is illustrated in Fig.~\ref{f:nE}. 
The component along $\bm{n}_r$ represents the horizontal polarization component, while the components along $\bm{n}_{r}^{\perp}$ and $\bm{n}_3$ represent the vertical polarization components.

\begin{figure}[!t]
    \centering
    \includegraphics[width=.8\linewidth]{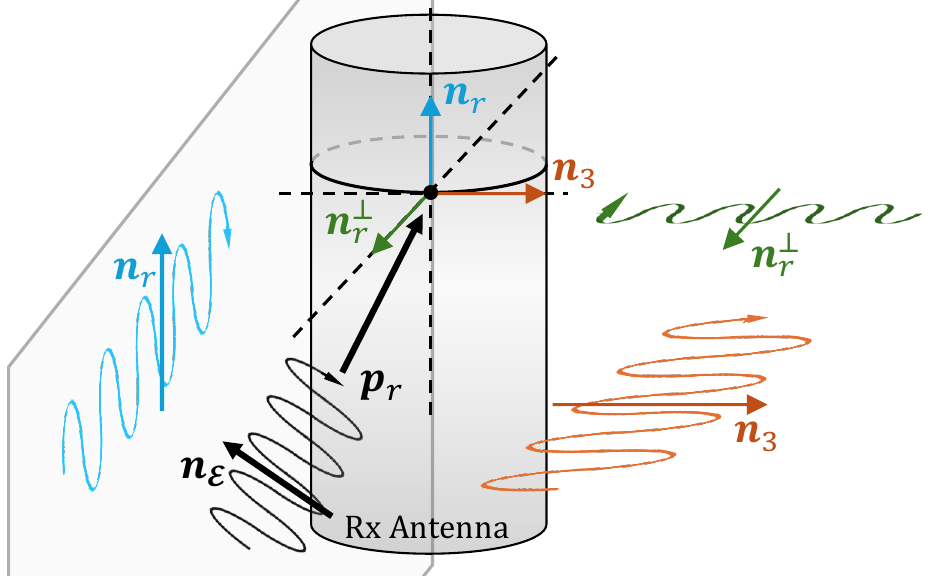}
    \caption{The electromagnetic field at the receiving antenna can be decomposed into three orthogonal components: along the direction of the receiving antenna $\bm{n}_r$, the normal direction of the incident signal $\bm{n}_{r}^{\perp}$, and a third direction $\bm{n}_3$ that is perpendicular to both $\bm{n}_r$ and $\bm{n}_{r}^{\perp}$.}
    \label{f:nE}
\end{figure}

According to \eqref{e:Gamma_para} and \eqref{e:Gamma_perp}, we have $\Gamma_{\parallel}\geq\Gamma_{\perp}$, meaning that the horizontal polarization component has a smaller reflection coefficient, and thus a better matching efficiency, than the vertical component. Therefore, increasing the component along $\bm{n}_r$ by adjusting antenna orientations enhances the channel gain.
However, achieving complete horizontal polarization requires the polarization matching angle $\alpha$ to be zero, which also corresponds to an incident angle $\theta^{(i)}$ of zero. Under this condition, the reflection coefficient $\Gamma_{\parallel}$ increases, reducing the matching efficiency. Therefore, simply minimizing $\alpha$ is suboptimal. Instead, we need to balance the horizontal and vertical polarization components to maximize the overall matching efficiency.

Since the polarization direction $\bm{n}_{\bm{\mathcal{E}}}$ is perpendicular to the propagation direction, and $\bm{p}_r$ is perpendicular to the antenna direction, the polarization matching angle satisfies  $\alpha \allowbreak =\langle \bm{n}_{\bm{\mathcal{E}}},\allowbreak \bm{n}_{r}\rangle \allowbreak \geq \theta^{(i)} =\allowbreak  \langle \bm{n}_{r}^{\perp},\allowbreak  \bm{p}_r\rangle$.
To increase the incident angle and reduce the reflection coefficient without introducing additional vertical polarization components along $\bm{n}_3$, it is optimal to maintain $\alpha = \theta^{(i)}$. For linearly polarized antennas, this configuration implies that the transmitting antenna, the receiving antenna, and the propagation path are all aligned within the same plane, confirming Proposition~\ref{prop:4}. It is important to note that lying in the same plane does not necessarily mean that the transmitting and receiving antennas are parallel, as depicted in Figs.~\ref{f:P2P_nt}(a) and \ref{f:P2P_nr}(a).


Figs.~\ref{f:P2P_nt}(b) and \ref{f:P2P_nr}(b) demonstrate that the worst-case orientation consistently occurs when one antenna is directed toward the other. This phenomenon is explained by:
\begin{itemize}
    \item For the transmitting antenna in Figs.~\ref{f:P2P_nt}(b), when the emission angle $\theta^{(e)}=0$, meaning the transmitting antenna points directly at the receiver, the electric field intensity in that direction is effectively zero due to the dipole antenna's radiation pattern.
    \item For the receiving antenna in Figs. \ref{f:P2P_nr}(b), when the incident angle $\theta^{(i)} = \frac{\pi}{2}$, the reflection coefficients $\Gamma_{\parallel, 1}=\Gamma_{\perp, 1}=1$, indicating total reflection with no energy entering the antenna. Additionally, the polarization matching angle $\alpha = \frac{\pi}{2}$, indicating complete polarization mismatch and resulting in zero capturable energy.
\end{itemize}





%% file: SecV.tex
By incorporating polarization, we have redefined movable antennas and established more refined, polarization-aware channel models. This enhancement serves as a foundation for further optimizations, specifically tailored to harness the dynamic adjustments of antenna orientations and positions to maximize system performance. In this paper, we focus on the fundamental sum-rate optimization \cite{tse2005fundamentals} to reveal the potential of PAMA.  

\subsection{Sum-rate optimization}\label{sec:optimization}
We consider a multi-user multiple input single output (MU-MISO) configuration where the BS serves $K$ users simultaneously, $K\leq L$, using transmit beamforming.
Let the source data symbols intended for the users be $ \bm{x} = [x_1, x_2, \dots, x_K]^\top $, where the symbol power is normalized such that $ \mathbb{E}[|x_\kappa|^2] = 1 $, $\kappa=1,2,\dots,K$. The beamforming vector for the $\kappa$-th user is denoted by $ \bm{w}_\kappa \in \mathbb{C}^{L \times 1} $, with individual elements $ w_{\ell,\kappa} $. A power scaling matrix, defined as $ \bm{P} = \text{diag}\{P_{1}, P_{2}, \dots, P_K \} $, is employed to control the power distribution among the users. After beamforming, the signal transmitted from the BS antennas, denoted by $ \bm{s}\triangleq [s_1, s_2, \dots, s_L]^\top $, is written as
\begin{eqnarray} \label{e:s}
s_\ell = \sum_{\kappa=1}^{K} \sqrt{P_{ \kappa}} w_{\ell,\kappa} x_\kappa,
\end{eqnarray}
subject to the constraint $ \sum_{\kappa=1}^{K}  P_{ \kappa} = P$, where $P$ denotes the total power constraint.

The symbols received by the users, denoted by $\bm{y} \triangleq [y_1, y_2, \dots, y_K]^\top$, are given by
\begin{eqnarray} \label{e:transmission}
\bm{y} = \bm{H} \bm{W} \bm{P} \bm{x} + \bm{z}, 
\end{eqnarray}
where $\bm{H} = [\bm{h}_1^\top, \bm{h}_2^\top, \dots, \bm{h}_K^{\top}]^\top$ is the channel matrix with the elements defined by \eqref{e:h}; $\bm{W}\triangleq[\bm{w}_1,\bm{w}_2,...,\bm{w}_K]$ is the precoding matrix; $\bm{z} = [z_1, z_2, \dots, z_K]^\top$ is the additive white Gaussian noise (AWGN) vector, where $z_k\sim\mathcal{CN}(0,\sigma^2_{\kappa})$.

At each user  $\kappa$, the SINR is defined as
\begin{eqnarray} \label{e:SINR_k}
&&\hspace{-1cm} \gamma_\kappa = \frac{P_\kappa \left|\bm{h}_\kappa \bm{w}_\kappa   \right|^2}{ \sigma_\kappa^2 + \sum_{j = 1, j\neq \kappa}^K{\left(P_j \big|\bm{h}_\kappa \bm{w}_j \big|^2 \right)} },
\end{eqnarray}
and the achievable rate is
\begin{eqnarray} \label{e:Rk}
R_\kappa = \frac{1}{2} \log_2{\left( 1 + \gamma_\kappa \right)}.
\end{eqnarray}

The average achievable rate across all users is
\begin{equation} \label{e:R}
R = \frac{1}{K} \sum\limits_{k = 1}^K {R_\kappa}
\triangleq \frac{1}{2} \log_2{\left( 1 + \gamma_{\text{total}} \right)} ,
\end{equation}
with the equivalent total SINR for the system:
\begin{equation}\label{e:SNR}
\gamma_{\text{total}} = \left[ \prod_{\kappa = 1}^{K} \left( 1 + \gamma_\kappa \right) \right]^{\frac{1}{K}} - 1.
\end{equation}

Our goal is to maximize the equivalent total SINR of the system by jointly optimizing the positions $\bm{p}_{t, \ell}$ and orientations $\bm{n}_{t, \ell}$ of the transmitting antennas, the orientations $\bm{n}_{r, \kappa}$ of the receiving antennas, beamforming matrix $\bm{W}$, and the power scaling matrix $\bm{P}$. The optimization problem is formulated as
\begin{subequations}
\begin{align}
\text {(P1):}\quad&\max_{\bm{W}, \bm{P}, \{\bm{n}_{t, \ell}\}, \{\bm{p}_{t, \ell}\}, \{\bm{n}_{r, \kappa}\} } \gamma_{\text{total}},\label{e:P1} \\ 
{\text {s.t.}~} 
    & \sum_{\kappa=1}^{K}  P_{\kappa} = P \label{e:P1power}\\
    & \left| \bm{p}_{t, \ell} - \bm{p}_{t, i} \right| \leq \frac{\lambda}{2}, \quad \forall 1 \leq i \leq L, i \neq \ell, \label{e:P1ref}\\
    & \bm{p}_{t, \ell} \in \mathcal{C}_t. \label{e:P1region} 
\end{align}
\end{subequations} 
In constraint \eqref{e:P1ref}, a minimum separation of $\frac{\lambda}{2}$ between any two antennas is imposed  to prevent collisions. The term $\mathcal{C}_t$ in \eqref{e:P1region} represents the movable region of the transmitting antennas.

To solve (P1), we employ zero-forcing (ZF) beamforming, which can fully eliminate the inter-channel interferences \cite{tervo2015optimal}. The ZF beamformer is given by 
\begin{equation*} \label{e:w_prime}
\bm{W}^\prime 
\triangleq [\bm{w}_1^\prime, \bm{w}_2^\prime, \dots, \bm{w}_K^\prime] 
= \bm{H}^H \left( \bm{H} \bm{H}^H \right)^{-1},
\end{equation*}
\begin{equation} \label{e:w}
\bm{W} 
= \left[\frac{\bm{w}_1^\prime}{|\bm{w}_1^\prime|}, \frac{\bm{w}_2^\prime}{|\bm{w}_2^\prime|}, \dots, \frac{\bm{w}_K^\prime}{|\bm{w}_K^\prime|} \right] .
\end{equation}

Since ZF beamforming completely removes interferences, the power allocation matrix can be easily obtained by water-filling \cite{yu2004iterative}, yielding
\begin{eqnarray} \label{e:P}
P_{\kappa} = \max{\left\{ \gamma_{th} - \frac{1}{\gamma_\kappa}, ~0 \right\}},
\end{eqnarray}
where $\gamma_{th}$ is the water level that can be determined iteratively to satisfy the total power constraint.

\begin{figure}[!tb]
    \centering
    \includegraphics[width=0.45\linewidth]{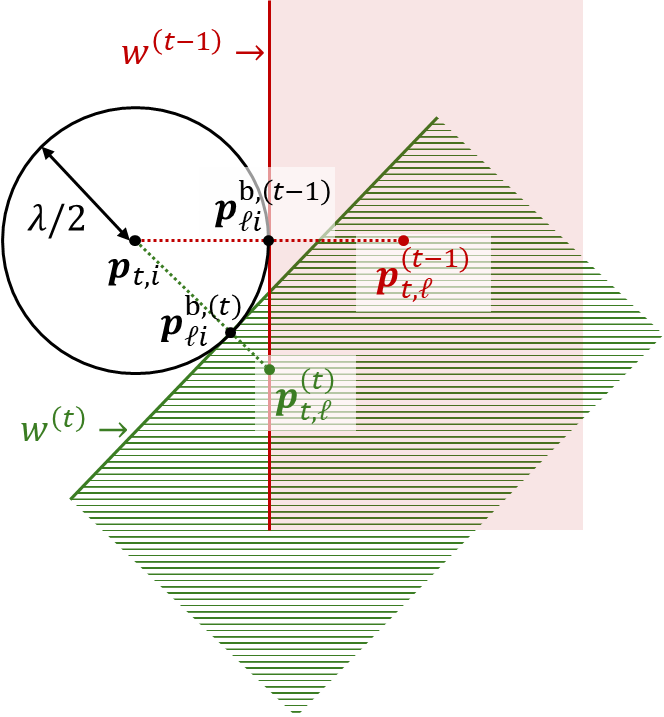}
    \caption{ A linear approximation of the non-convex constraint in \eqref{e:P1ref2} is obtained and projected onto the feasible set.}
    \label{fig:R_3_3}
\end{figure}

Next, we transform the non-convex feasible region of $\bm{p}_{t,\ell}$ into a convex set by linearizing the non-convex constraint in \eqref{e:P1ref} \cite{shao20246d}.
During the $ t $-th iteration of the gradient descent, we consider the position $ \bm{p}_{t, \ell}^{(t-1)} $ from the previous iteration and the position $ \bm{p}_{t, i} $. As part of the linearization of the non-convex constraint, the boundary point is first computed as
\begin{equation} \label{e:p_lt_b}
\bm{p}_{\ell i}^{\text{b}, (t-1)} = \bm{p}_{t, i} + \frac{\lambda}{2} \frac{\bm{p}_{t, \ell}^{(t-1)} - \bm{p}_{t, i}}{\left| \bm{p}_{t, \ell}^{(t-1)} - \bm{p}_{t, i} \right|}.
\end{equation}

The boundary of the hyperplane is then defined w.r.t. $\bm{p}_{\ell i}^{\text{b}, (t-1)}$ and the vector $\bm{p}_{t, \ell}^{(t-1)} - \bm{p}_{t, i}$ as
\begin{equation}
\left\{ \bm{w}^{(t-1)}_c \left| \left( \bm{p}_{\ell i}^{\text{b}, (t-1)} - \bm{p}_{t, i}  \right)^\top \bm{w}^{ (t-1)} = c \right. \right\},
\end{equation}
and the corresponding closed halfspace divided by this hyperplane is given by
\begin{equation} \label{e:halfspace}
\left\{\bm{w}^{(t-1)} \left| \left( \bm{p}_{t, \ell}^{(t-1)} - \bm{p}_{t, i} \right)^{\!\!\! \top} \!\!\! \left(\bm{w}^{(t-1)} - \bm{p}_{\ell i}^{\text{b}, (t-1)} \right) \geq 0 \right.  \right\},
\end{equation}
which includes $\bm{p}_{\ell i}^{\text{b}, (t-1)}$ and any vector that forms an acute angle with the vector $\bm{p}_{t, \ell}^{(t-1)} - \bm{p}_{t, i}$, as illustrated in Fig.~\ref{fig:R_3_3}.

By setting $\bm{w}^{(t-1)} = \bm{p}_{t, \ell}^{(t)}$ in \eqref{e:halfspace}, the non-convex constraint can be approximated as the following linear inequality:
\begin{eqnarray} \label{e:P1ref2}
\left( \bm{p}_{t, \ell}^{(t-1)} \!\!\! - \! \bm{p}_{t, i} \right)^{\!\!\!\top} \!\!\! \left(\bm{p}_{t, \ell}^{(t)} \!\!\! - \! \bm{p}_{\ell i}^{\text{b},(t-1)} \right) \! \geq\! 0, ~ \forall 1 \leq i \leq L, i \neq \ell.
\end{eqnarray}

Thus, problem (P1) is redefined to a more tractable form
\begin{subequations}\label{e:P2}
\begin{align}
\text {(P2):}\quad&\max_{\{\bm{n}_{t, \ell}\}, \{\bm{p}_{t, \ell}\}, \{\bm{n}_{r, \kappa} \}} \gamma_{\text{total}},\label{e:P1-1} \\ 
{\text {s.t.}~} 
    & \eqref{e:w}, ~ \eqref{e:P}, ~ \eqref{e:P1power}, ~ \eqref{e:P1ref2}, ~ \eqref{e:P1region}. \label{e:P1-1ref}
\end{align} 
\end{subequations}

\begin{algorithm}[!tb]
\caption{ Alternating Optimization for Solving (P2).}
\label{alg:AO}
\begin{algorithmic}[1] 
\REQUIRE ~~$f$, $\sigma^2$, $P$, $\varepsilon_r$, $\mu$, $A_F$, $\mathcal{C}_t$, $\bm{p}_{r, \kappa}$.
\STATE {Initialize: $t = 0$, $\bm{p}_{t, \ell}$, $\bm{n}_{t, \ell}$, $\bm{n}_{r, \kappa}$}.
\REPEAT 
\STATE {Update $\bm{n}_{t, \ell}$  by radient descent.}
\STATE {Update $\bm{p}_{t, \ell}$  by projected radient descent.}
\STATE {Update $\bm{n}_{r, \kappa}$  by radient descent.}
\STATE {Update $\bm{W}$ based on (27).}
\STATE {Update $\bm{P}$ based on (28).}
\STATE {Update $ \gamma_{\text{total}}$ based on (25).}
\UNTIL {Convergence or $t = t_{\text{max}}$.}
\end{algorithmic}
\end{algorithm}

Based on Problem (P2), we derive the gradient of the equivalent total SNR $ \gamma_{\text{total}} $ with respect to each optimization variable. In particular, the gradient with respect to the direction vector can be computed by leveraging the relationship between the antenna's polar and azimuthal angles defined in \eqref{e:nt}. Taking $ \bm{n}_{t, \ell} $ as an example, the gradient of the objective with respect to the angular parameters is given by
\begin{eqnarray}
    \nabla_{\theta^{(n)}_{t, \ell}}\gamma_{\text{total}} \!\!&=& \!\!\frac{\gamma_{\text{total}}\left( \theta^{(n)}_{t, \ell} + \Delta \theta^{(n)}_{t, \ell} \right) - \gamma_{\text{total}}\left( \theta^{(n)}_{t, \ell}\right)}{\Delta \theta^{(n)}_{t, \ell}}, \notag\\
    \nabla_{\phi^{(n)}_{t, \ell}}\gamma_{\text{total}} \!\!&=& \!\! \frac{\gamma_{\text{total}}\left( \phi^{(n)}_{t, \ell} + \Delta \phi^{(n)}_{t, \ell} \right) - \gamma_{\text{total}}\left( \phi^{(n)}_{t, \ell}\right)}{\Delta \phi^{(n)}_{t, \ell}}.
\end{eqnarray}
Given that $ \bm{n}_{t, \ell} $ is a periodic function of the angles $ \theta^{(n)}_{t, \ell} $ and $ \phi^{(n)}_{t, \ell} $, the angles are projected onto their respective unit periodic domains after each iteration to ensure validity while preserving the value of $ \bm{n}_{t, \ell} $. Subsequently, $ \bm{n}_{t, \ell} $ is updated based on the computed gradient.

Similarly, the gradient w.r.t $ \bm{p}_{t, \ell} $ can be derived as
\begin{equation}
    \nabla_{\bm{p}_{t, \ell}}\gamma_{\text{total}} = \frac{\gamma_{\text{total}}\left( \bm{p}_{t, \ell} + \Delta \bm{p}_{t, \ell} \right) - \gamma_{\text{total}}\left( \bm{p}_{t, \ell}\right)}{\Delta \bm{p}_{t, \ell}},
\end{equation}
After updating $ \bm{p}_{t, \ell} $ using the computed gradient, it is projected onto the feasible set according to equation (32).

Algorithm~\ref{alg:AO} summary the overall optimization process.



\begin{rem}[Channel estimation in PAMA]
Channel estimation poses a significant challenge in MA systems because determining the optimal antenna position requires knowledge of the channel characteristics at every possible position and orientation within the movement region $\mathcal{C}_t$. This task is particularly complex in radio frequency bands with rich multipath environments.

However, within our PAMA framework operating in the mmWave band, channel estimation is considerably simplified due to the dominance of the LOS path. PAMA only requires information about user positions and antenna orientations. With this information, we can accurately infer the channel gains at any position and orientation using \eqref{e:h}.
Specifically, the impact of estimation errors on the SNR performance of the PAMA system with $K = 4$ is analyzed in Appendix \ref{sec:AppC}.
\end{rem}

\begin{rem}[The $K> L$ case]
ZF beamforming effectively eliminates inter-user interference only when the number of users $K$ does not exceed the number of transmit antennas $L$. When $K> L$, we can partition the users into groups of size up to $L$ and assign orthogonal resources to each group. Employing more advanced detection algorithms, such as GMMSE \cite{shi2011iteratively} and GPIP \cite{choi2019joint}, can further enhance SNR performance compared to the ZF baseline. A comparison of optimization algorithms is provided in Appendix \ref{sec:AppD}, where it is shown that regardless of the algorithm used, the PAMA system consistently outperforms its fixed-antenna counterpart due to the combined benefits of antenna mobility and polarization matching. Alternatively, exploring different optimization criteria, such as minimizing interference leakage \cite{sadek2007leakage}, can lead to alternative beamforming designs. 
\end{rem}

\subsection{Numerical simulations} \label{ana_mp2mp}

\begin{table}[t]
\centering
\caption{Simulation parameter settings.}
\begin{tabular}{ccc}
     \toprule
     \textbf{Symbol}   & \textbf{Description} & \textbf{Value} \\
     \midrule
     $f$ & Carrier frequency   & 30 GHz \\
     $\lambda$ & Wave length   & 0.01 m \\
     $\sigma^2$ & Noise power & -20 dBm \\
     $P$ & Signal power & 0.5 W \\
     $\varepsilon_r$ & Relative permittivity & 2 \\
     $\mu$ & Permeability & $4 \pi \times 10^{-7}$ \\
     $A_F$ & Antenna factor & 1 \\
     $\mathcal{C}_t$ & Transmitter region & $x_{r,\kappa},y_{r,\kappa},z_{r,\kappa}\in[-100\lambda, 100\lambda]$ \\
     $L$ & \#transmitting antennas & 8 \\
     \bottomrule
\end{tabular}
\label{tab:sim_para}
\end{table}


\begin{figure*}[t]
\centering
    \begin{subfigure}[!t]{0.245\linewidth}
        \centering
        \includegraphics[width=1\linewidth]{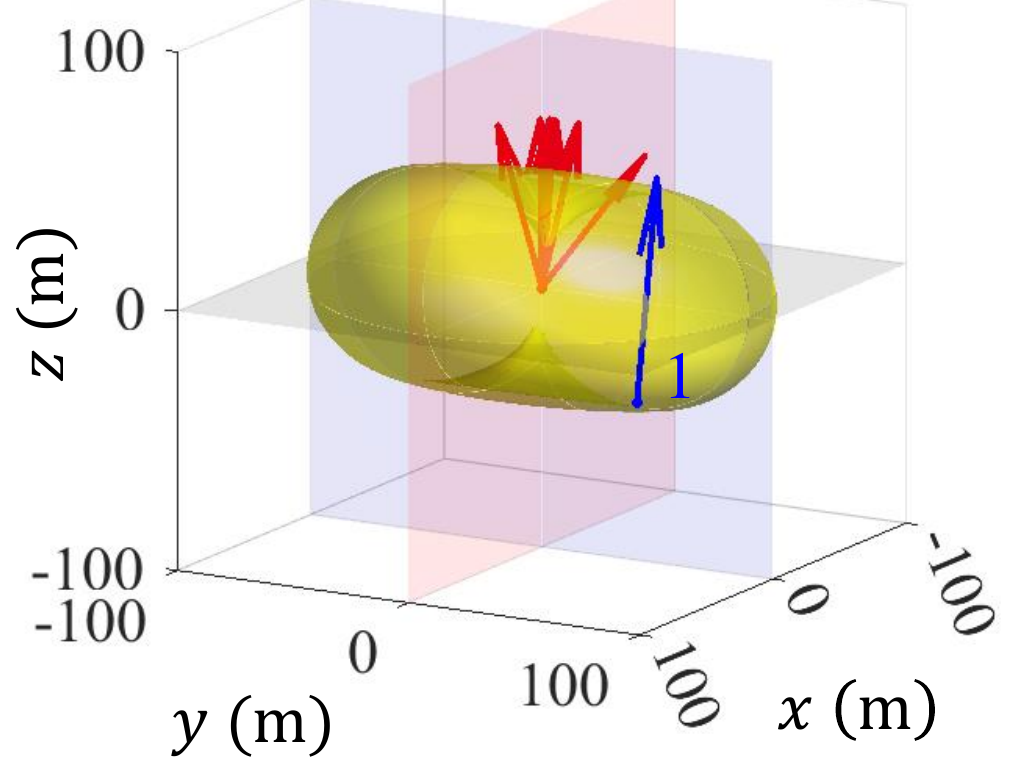}
        \caption{$K = 1$}\label{f:MIMO_EF_81}
    \end{subfigure}
    \begin{subfigure}[!t]{0.245\linewidth}
        \includegraphics[width=1\linewidth]{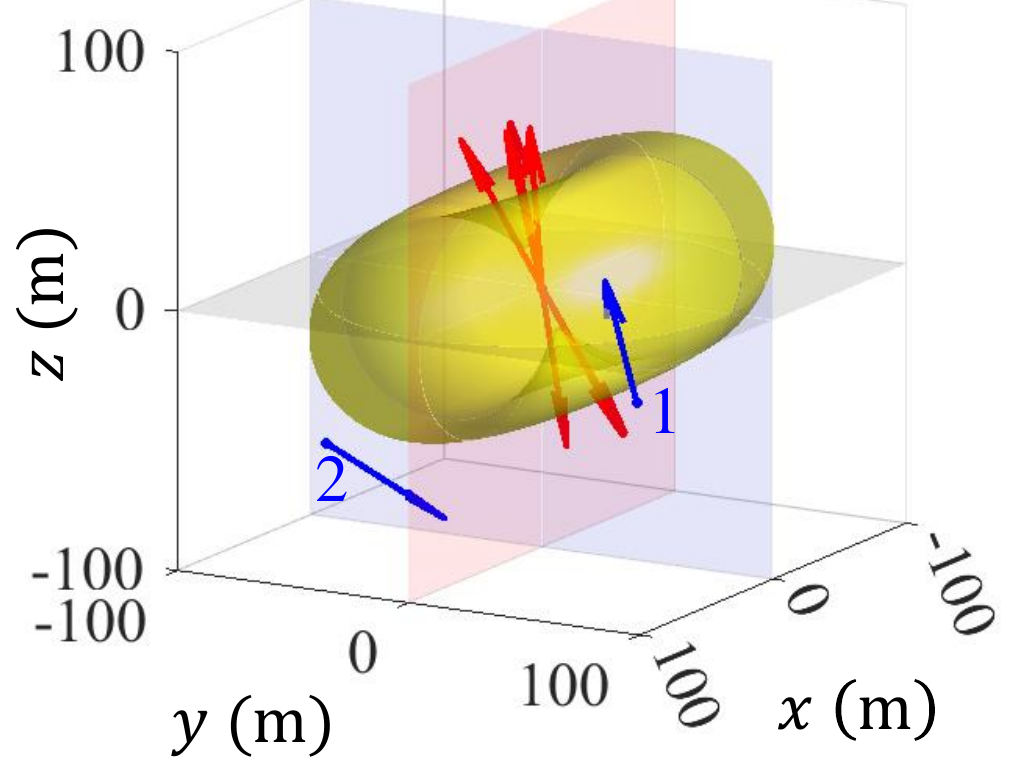}
        \caption{$K = 2$}\label{f:MIMO_EF_82}
    \end{subfigure}
    \begin{subfigure}[!t]{0.245\linewidth}
        \includegraphics[width=1\linewidth]{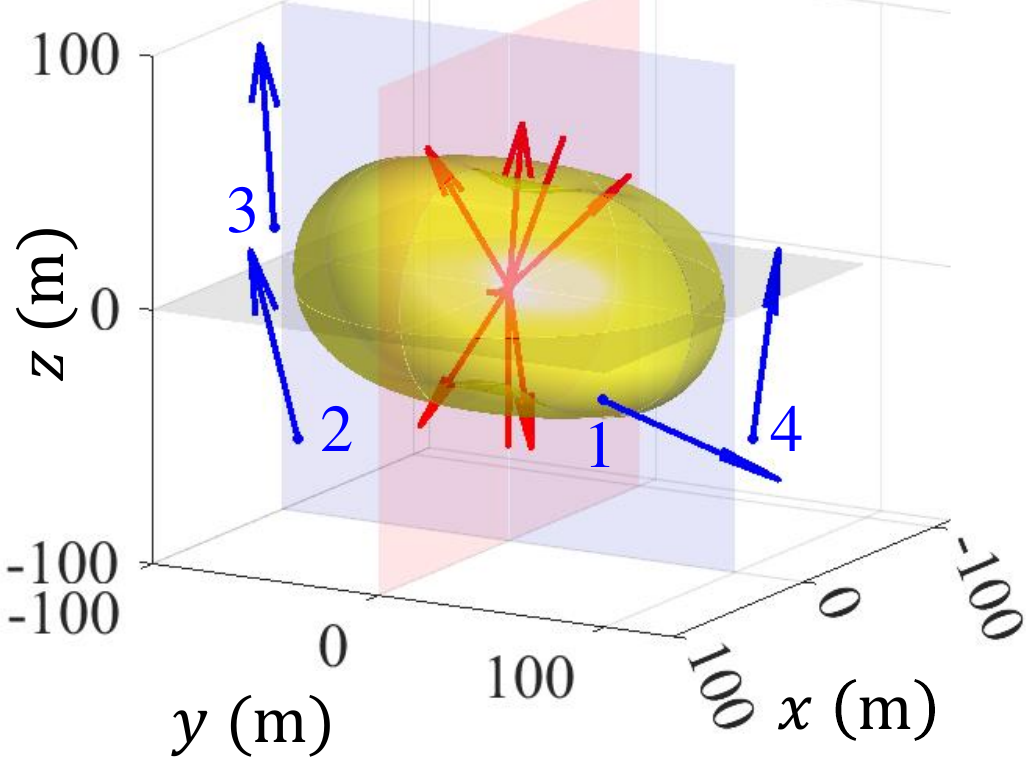}
        \caption{$K = 4$}\label{f:MIMO_EF_84}
    \end{subfigure}
    \begin{subfigure}[!t]{0.245\linewidth}
        \includegraphics[width=1\linewidth]{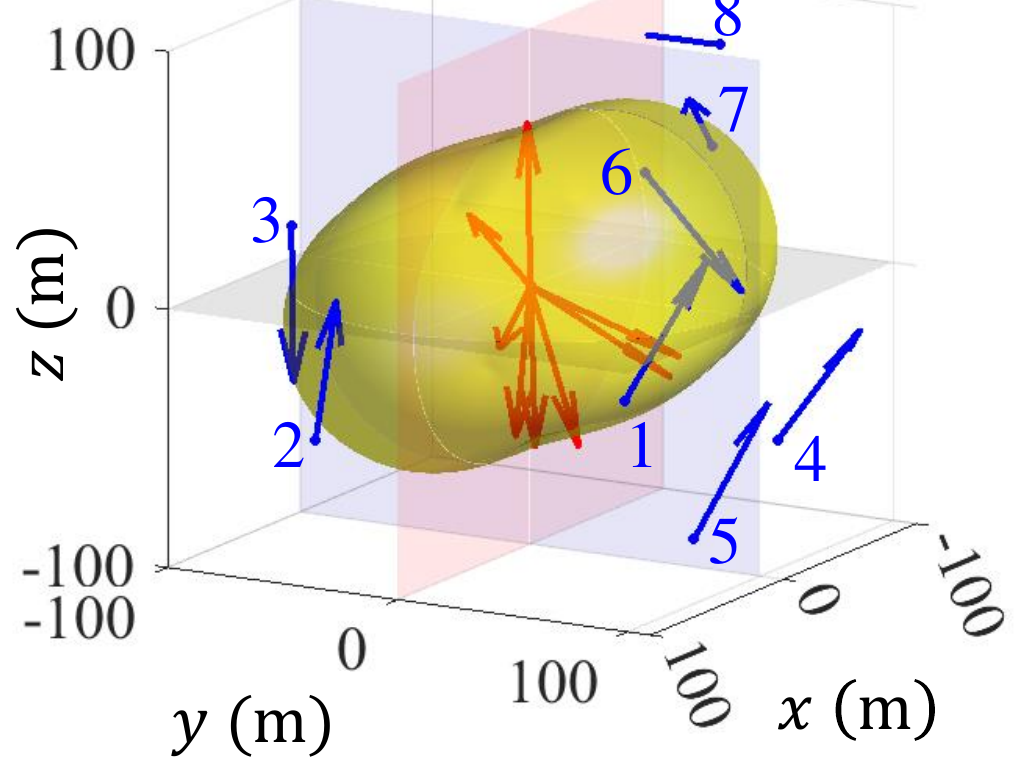}
        \caption{$K = 8$}\label{f:MIMO_EF_88}
    \end{subfigure}
\caption{Optimal positions and orientations of the transmitting antennas, optimal orientations of the receiving antennas, and the corresponding electric field isosurfaces in the MU-MISO system for different user counts.}
\label{f:isosurface}
\end{figure*}

This section presents numerical simulations to evaluate the performance of our PAMA framework in the context of sum-rate optimization. In the simulations, the system operate in 30 GHz mmWave and consider users uniformly distributed within a 3D coverage area -- a cube with dimensions of 200 meters $\times$ 200 meters $\times$ 200 meters -- centered at the origin $\bm{o}$ of the CCS. Unless otherwise specified, the transmission power is set to 0.5 W, and all simulations are conducted with the number of transmit antennas $L$ set to 8. A comprehensive summary of the simulation parameters is provided in Tab.~\ref{tab:sim_para}.

\subsubsection{Impact of polarization matching}
Our PAMA framework integrates polarization effects into MAs and reveals a significant advantage of MAs: the ability to achieve optimal polarization alignment between the transmitter and receiver. To validate this capability, we consider the following antenna configurations and channel model:
\begin{enumerate}[leftmargin=0.5cm]
    \item \textbf{Fixed vertically polarized antenna array}: A vertically polarized antenna array with elements spaced at half-wavelength intervals, based on a traditional channel model.
    \item \textbf{Fixed cross-polarized antenna array}: An antenna array arranged in a cross-polarized configuration as described in \cite{3GPP5G}, with elements separated by half a wavelength, based on a traditional channel model.
    \item \textbf{Continuous fluid antenna system} (CFAS): A vertically polarized antenna array that moves continuously within vertically placed tubes, with adjacent tubes spaced half a wavelength apart. During the antenna fluiding, only phase changes are considered.
    \item \textbf{Polarization-Ignorant 3D MA system}  (PI-3DMA): Transmitting antennas have optimized positions but vertical orientations; receiving antennas have random orientations. During the antenna translation, only phase changes are considered.
    \item \textbf{Polarization-Ignorant rotational MA system}   (PI-RoMA): Transceiver antennas have optimized orientations but fixed positions. During antenna rotation, only changes in the radiation pattern are considered.
    \item \textbf{Polarization-Ignorant 6D MA system}   (PI-6DMA): Both the positions and orientations of the transceiver antennas are optimized. During antenna movement, changes in phase and radiation pattern are considered.
    \item \textbf{Polarization-Aware 3D MA system}   (PA-3DMA): Transmitting antennas have optimized positions but vertically orientations; receiving antennas have random orientations. Applying the PAMA channel model.
    \item \textbf{Polarization-Aware rotational MA system}   (PA-RoMA): Transeiver antennas have optimized orientations but fixed positions. Applying the PAMA channel model.
    \item \textbf{PAMA}: Both the positions and orientations of the transeiver antennas are optimized. Applying the PAMA channel model.
\end{enumerate}

In Configurations 1 and 2, only fixed antennas with predefined positions and orientations are considered, without accounting for movement effects. Configurations 3 - 6 explore different types of MA while optimizing the system without considering polarization effects. In contrast, Configurations 7 - 9 integrate polarization matching into the channel model optimization. Notably, Configurations 4 - 6 and 7 - 9 share the same antenna structure, differing only in that Configurations 4 - 6 omit polarization effects in their channel models.
We compare system performance under different antenna configurations and channel models in Figs.~\ref{f:SNR_K1} and \ref{f:SNR_K2}.

\begin{figure}[t]
    \centering
    \includegraphics[width=0.75\linewidth]{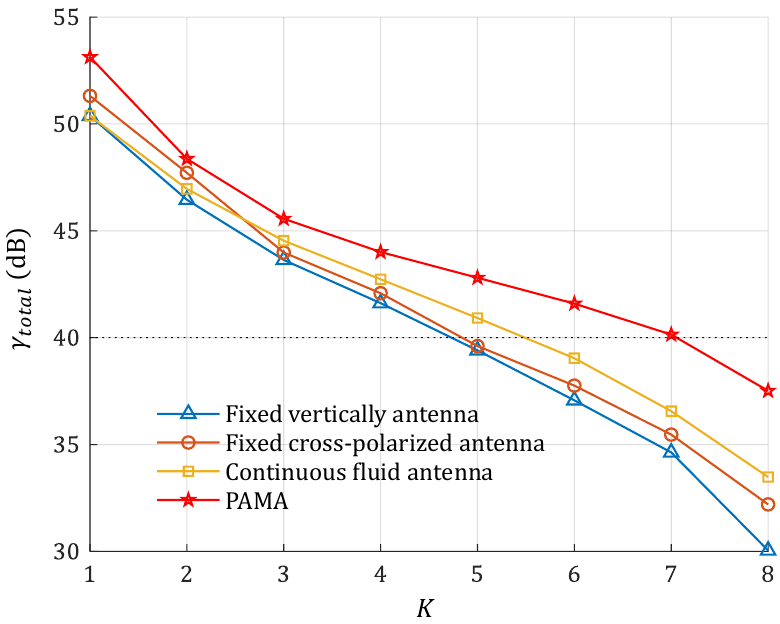}
    \caption{The equivalent total SNR $\gamma_{\text{total}}$ versus the number of users $K$ for different antenna configurations.}
    \label{f:SNR_K1}
\end{figure}

\begin{figure}[t]
    \centering
    \includegraphics[width=0.75\linewidth]{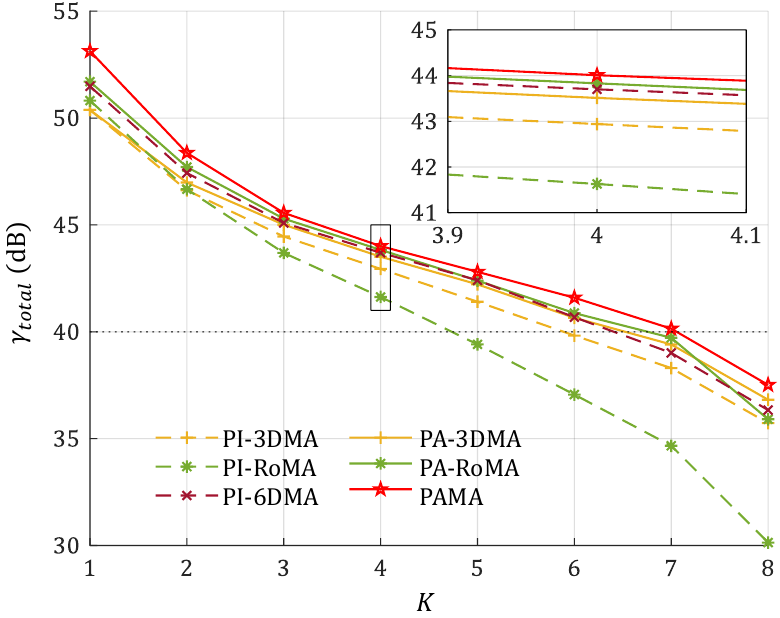}
    \caption{The equivalent total SNR $\gamma_{\text{total}}$ versus the number of users $K$ for different channel model.}
    \label{f:SNR_K2}
\end{figure}

As shown in Fig. \ref{f:SNR_K1}, the fixed cross-polarized antenna array achieves a slightly higher equivalent total SNR than the vertically polarized antenna array. This is because cross-polarization provides a dual-polarized configuration, mitigating the impact of polarization mismatch over a wider range. CFAS outperforms the cross-polarized antenna array due to the introduction of spatial degrees of freedom.

Fig. \ref{f:SNR_K2} further illustrates that for the same antenna configurations -- 3DMA, RoMA, and 6DMA -- systems that consider polarization matching efficiency can more effectively exploit the available degrees of freedom, leading to improved performance.
As shown in Fig.~\ref{f:SNR_K2}, while PI-6DMA appears to achieve slightly worse average performance, this is largely due to the rarity of severe polarization mismatch as disscussed in Sec.~ \ref{SecIV}. Fig.~\ref{f:1_3_2} reveals the key difference: under polarization mismatch, PI-6DMA suffers from limited gain, while PAMA maintains stable and effective optimization, fully leveraging its channel model and the additional degrees of freedom to realign polarization and recover performance.
These results validate the superior resilience and optimization capability of PAMA, particularly in non-ideal conditions where simplified models fail. The benefits brought by the refined channel model are thus not only theoretical but also practically significant, enabling reliable and high-performance communication even under adverse polarization conditions.

\begin{figure}[t]
    \centering
    \includegraphics[width=0.75\linewidth]{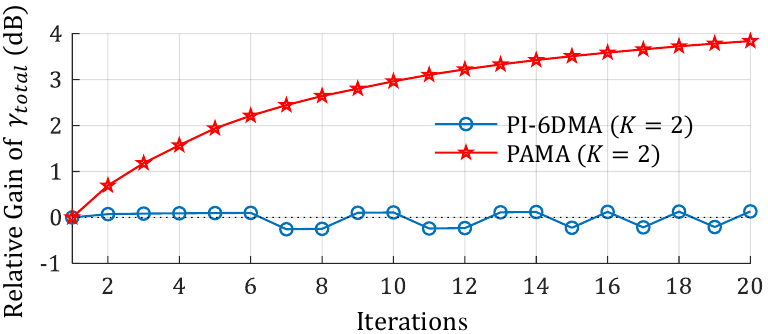}
    \caption{Comparison of relative gains between PAMA and PI-6DMA system  with $K = 2$.}
    \label{f:1_3_2}
\end{figure}

To provide a clearer understanding of the optimized PAMA in configuration 5, we present in Fig.~\ref{f:isosurface} the optimal orientations of both the transmitting and receiving antennas as the number of users increases, along with the corresponding electric field isosurfaces.
\begin{itemize}[leftmargin=0.5cm]
\item In the single-user scenario illustrated in Fig.~\ref{f:isosurface}(a), all transmitting antennas are oriented to focus their maximum radiation patterns toward the target. The combined electric field distribution from the eight antennas resembles that of a single half-wavelength dipole antenna, forming a characteristic doughnut-shaped pattern.
\item As the number of users increases, the electric field distribution becomes more isotropic to effectively accommodate all users' requirements.
\end{itemize}

\setcounter{figure}{10}
\begin{figure}[t]
    \centering
    \includegraphics[width=0.75\linewidth]{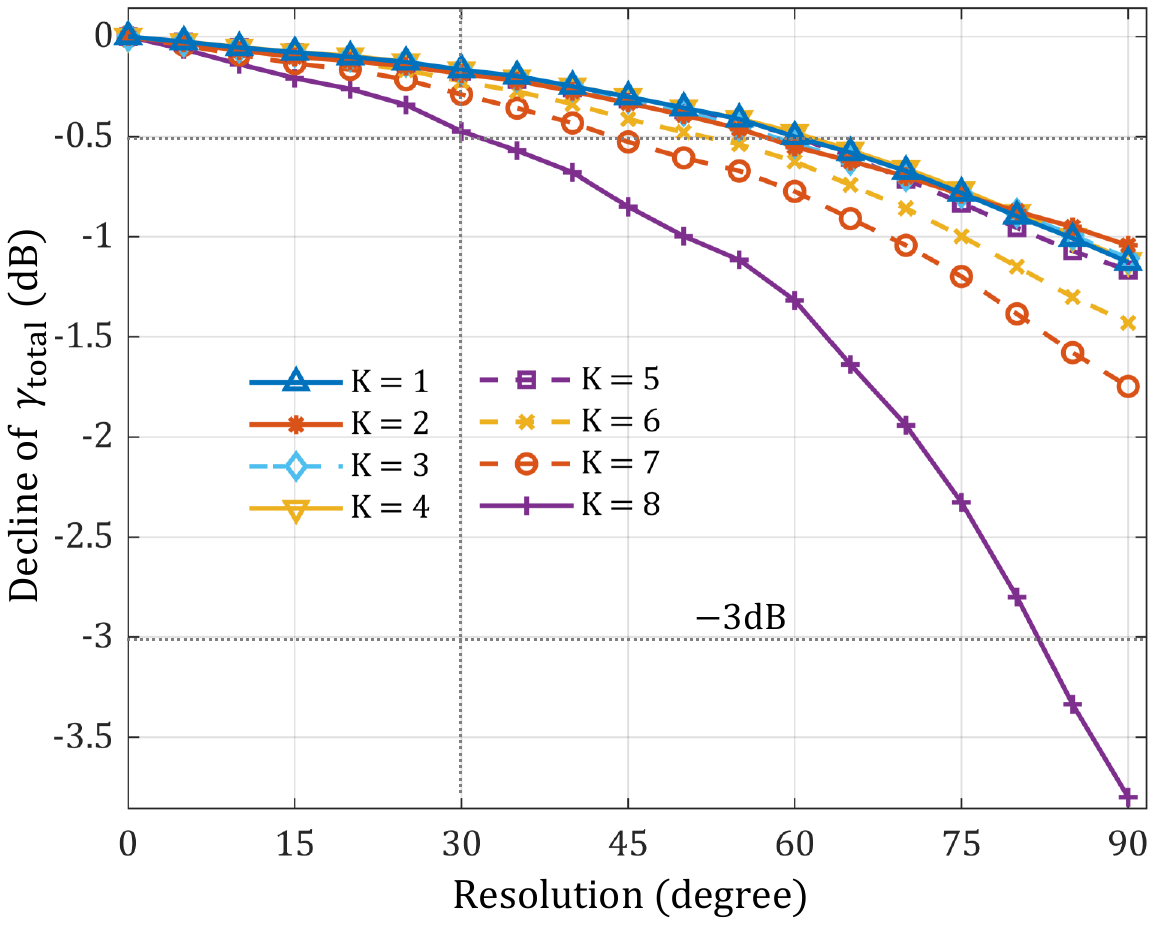}
    \caption{The equivalent total SNR $\gamma_{\text{total}}$ versus the antenna rotation granularity for varying numbers of users $K$.}
    \label{f:SNR_reso}
\end{figure}

\begin{figure}[t]
    \centering
    \includegraphics[width=0.75\linewidth]{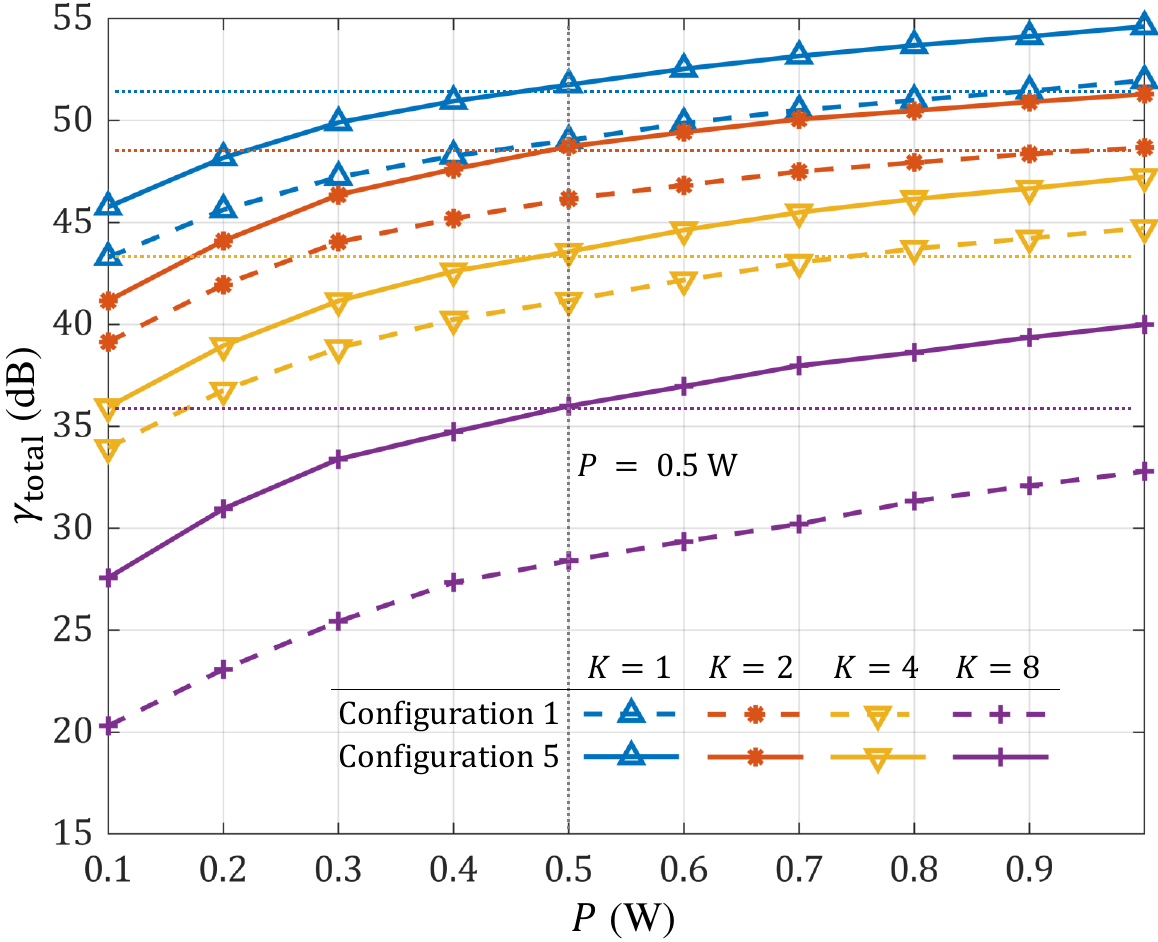}
    \caption{The equivalent total SNR $\gamma_{\text{total}}$ versus the total transmission power $P$ for varying numbers of users $K$.}
    \label{f:R_Ps}
\end{figure}

\subsubsection{Impact of antenna rotation granularity}
Implementing PAMA in practice faces the challenge of precisely controlling the rotations of both transmitting and receiving antennas. While the simulation results in Fig.~\ref{f:SNR_K2} assume infinite control precision, real-world systems are constrained by finite rotation granularity. 
To assess how rotation granularity affects system performance, we quantize the polar and azimuthal angles of the transmitting and receiving antennas, i.e., $\theta_{t,\ell}^{(n)}$, $\phi_{t,\ell}^{(n)}$, $\theta_{r,\kappa}^{(n)}$, and $\phi_{r,\kappa}^{(n)}$. We then evaluate the decline in the equivalent total SNR, $\gamma_{\text{total}}$, across different quantization levels. The simulation results are depicted in Fig.~\ref{f:SNR_reso}, where a resolution level of $0$ corresponds to the optimal orientation with infinite precision.

The findings show that the negative impact of reduced quantization precision on $\gamma_{\text{total}}$ becomes more pronounced as the number of users increases. For example, with a rotation precision limited to $80^\circ$, a system serving $K = 8$ users experiences a $3$ dB decrease in $\gamma_{\text{total}}$. When the rotation resolution is improved to $30^\circ$, the achievable $\gamma_{\text{total}}$ stays within $0.5$ dB of the maximum possible gain. In this scenario, each angular dimension -- both the polar and azimuthal angles -- requires adjustments in only $360^\circ/30^\circ=12$ discrete directions. This significantly reduces the implementation complexity, as controlling antenna rotations becomes much more practical with a finite set of orientations. These results demonstrate that PAMA can perform effectively even with finite control precision, confirming that high-resolution rotational control is not necessary to achieve near-optimal performance.

\subsubsection{Transmit power reduction and convergence behavior}
Next, we examine how the equivalent total SNR $\gamma_{\text{total}}$ varies with the total transmit power $P$. The simulation results are shown in Fig.~\ref{f:R_Ps} and explained as follows.
\begin{itemize}[leftmargin=0.5cm]
\item For a fixed number of users, increasing $P$ leads to a higher $\gamma_{\text{total}}$, whereas increasing the number of users results in a lower $\gamma_{\text{total}}$.
\item With polarization matching through rotations of both transmitting and receiving antennas, the PAMA framework demonstrates greater power efficiency. To achieve the same total SNR, PAMA significantly reduces the required transmit power. This outcome aligns with expectations, as the transmitter can adjust its orientation to direct higher field intensity toward each user, while the receiver can optimize matching efficiency by rotating its antenna for optimal alignment.
\end{itemize}

\begin{figure}[!t]
    \centering
    \includegraphics[width=0.75\linewidth]{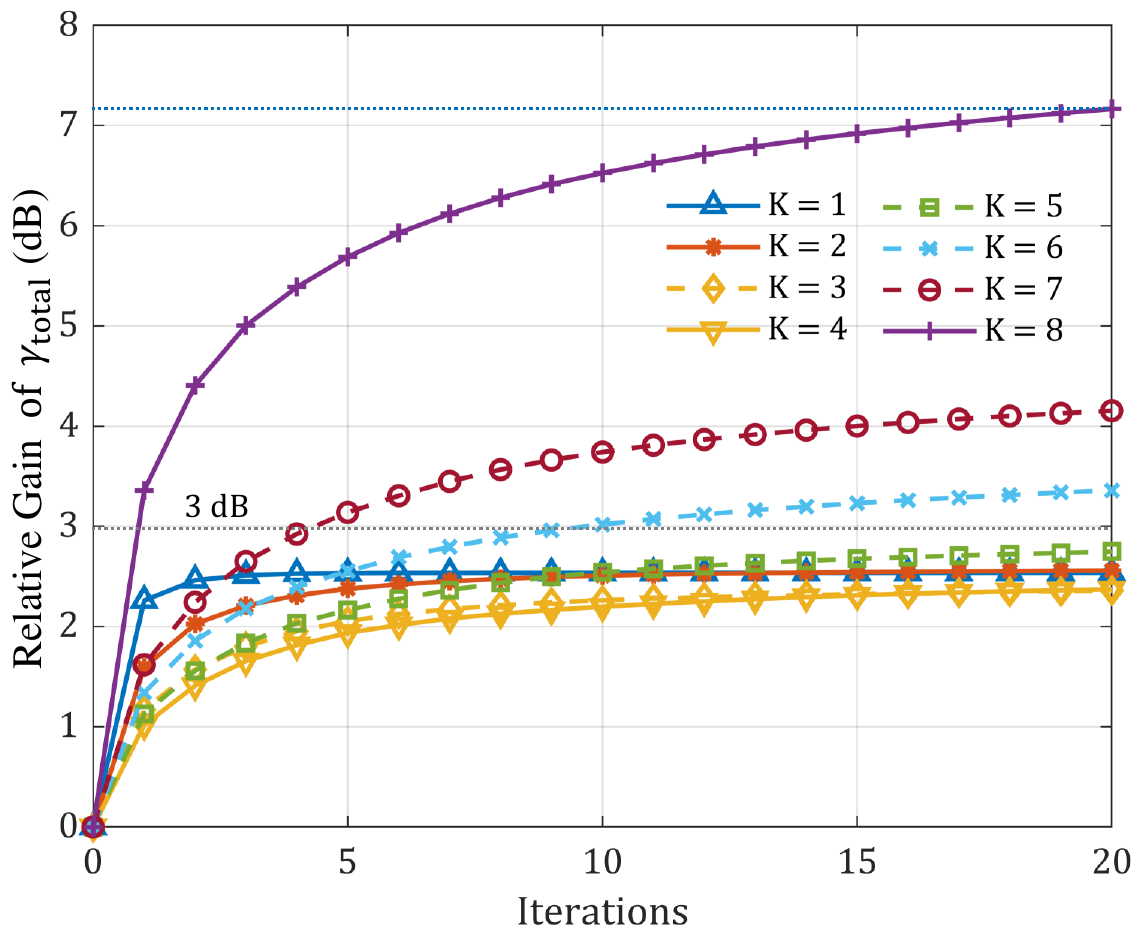}
    \caption{Comparison of relative gains between PAMA and fixed vertically polarized antenna array with the increase in $K$, and the convergence behavior of the optimization algorithm.}
    \label{f:SNR_iter}
\end{figure}

Finally, Fig.~\ref{f:SNR_iter} illustrates the convergence behavior of our optimization algorithm by showing the relative gain in total SNR. The results reveal that:
\begin{itemize}[leftmargin=0.5cm]
\item The total SNR increases with each iteration across various values of $K$. For $K < 5$, the algorithm typically converges within about $10$ iterations, achieving a gain of less than $3$ dB. However, when $K\geq 5$, more iterations are required for convergence.
\item Despite the increased number of iterations needed for larger $K$, the additional receiving antennas introduce greater degrees of freedom, enabling higher gains within a reasonable number of iterations. For instance, with $K=8$, the gain can reach up to $7$ dB within 20 iterations.
\end{itemize}

%% file: AppendixA.tex
As a preliminary, the appendix explains the derivation of \eqref{e:E0}, the radiation patterns of a half-wave dipole.

Maxwell’s equations can be written as
\begin{eqnarray} 
&&\hspace{-1cm} \nabla \times \bm{\mathcal{E}} = - j \omega \mu \bm{\mathcal{H}}, \label{e:ME1}\\
&&\hspace{-1cm} \nabla \times \bm{\mathcal{H}} = j \omega \varepsilon \bm{\mathcal{E}} + \bm{J}, \label{e:ME2}\\
&&\hspace{-1cm} \nabla \cdot \bm{\mathcal{E}} = \frac{\rho}{\varepsilon}, \label{e:ME3}\\
&&\hspace{-1cm} \nabla \cdot \bm{\mathcal{H}} = 0, \label{e:ME4}
\end{eqnarray}
where
\begin{itemize}[leftmargin=0.45cm]
    \item $\bm{\mathcal{E}}$ and $\bm{\mathcal{H}}$ refer to the electric field and magnetic field.
    \item  $\mu$ is the magnetic permeability of air, describing its ability to conduct magnetic field lines. $\varepsilon$ is the permittivity, characterizing a material’s capacity to store electric charge.
    \item $\bm{J}$ and $\rho$ represent the source current density and source charge density, respectively.
    \item $\omega$ is the angular frequency of the wave. The expression $\bm{a} \times \bm{b}$ denotes the cross product; $\nabla \times \bm{b}$ denotes the curl of $\bm{b}$; $\bm{a} \cdot \bm{b}$ denotes the dot product; $\nabla \cdot \bm{b}$ denotes the divergence of $\bm{b}$, and $\nabla$ is the Hamiltonian operator.
\end{itemize} 

Since the divergence of $\bm{\mathcal{H}}$ in \eqref{e:ME4} is zero, the vector field $\bm{\mathcal{H}}$ has no sources or sinks; it only circulates. This allows us to define a magnetic vector potential $\bm{A}$ so that $\bm{\mathcal{H}}$ can be described as the curl of $\bm{A}$:
\begin{eqnarray} \label{e2.36}
\bm{\mathcal{H}} = \frac{1}{\mu} \nabla \times \bm{A} .
\end{eqnarray}

Substituting \eqref{e2.36} into \eqref{e:ME1} yields $\nabla \times ( \bm{\mathcal{E}} + j \omega \bm{A}) = 0$. In vector calculus, if the curl of a vector field is zero, that field can be expressed as the gradient of a scalar potential. Therefore, we can define an electric scalar potential $\Phi$ such that
\begin{eqnarray} \label{e2.38}
\bm{\mathcal{E}} + j \omega \bm{A} = - \nabla \Phi.
\end{eqnarray}
It is important to note that $\nabla \times \nabla \Phi$ is always zero for any scalar potential $\Phi$.

Equations \eqref{e2.36} and \eqref{e2.38} show how $\bm{\mathcal{H}}$ and $\bm{\mathcal{E}}$ can be represented in terms of their potential functions. This means that once we find these potential functions, we can determine the corresponding field expressions. The next step is to solve for these potential functions.

First, substituting \eqref{e2.36} into \eqref{e:ME2} gives
\begin{eqnarray} \label{e2.40}
\frac{1}{\mu} \nabla \times \nabla \times \bm{A} = j \omega \varepsilon \bm{\mathcal{E}} + \bm{J}.
\end{eqnarray}
Next, we use the vector identity $\nabla \times \nabla \times \bm{A} \equiv \nabla (\nabla \cdot \bm{A}) - \nabla^2 \bm{A}$, where $\nabla^2 = \nabla \cdot \nabla$ represents the Laplace operator, which calculates the divergence of the gradient. By applying \eqref{e2.38}, we can rewrite \eqref{e2.40} as
\begin{eqnarray} \label{e2.42}
\nabla (\nabla \cdot \bm{A}) - \nabla^2 \bm{A} = j \omega \mu \varepsilon \left( - j \omega \bm{A} - \nabla \Phi \right) + \mu \bm{J}.
\end{eqnarray}

After a simple rearrangement and using the Lorentz condition $\nabla \cdot \bm{A} = - j \omega \mu \varepsilon \Phi$\cite{huang2021antennas}, we get
\begin{eqnarray} \label{e2.45}
\nabla ^2 \bm{A} + \omega ^2 \mu \varepsilon \bm{A} = - \mu \bm{J}.
\end{eqnarray}
This equation is known as the vector wave equation. With this differential equation, if the current $\bm{J}$ is known, we can solve for the vector potential $\bm{A}$. Once $\bm{A}$ is determined, the magnetic field $\bm{\mathcal{H}}$ can be found using \eqref{e2.36}. Then, by using the Lorentz condition along with \eqref{e2.38}, the electric field $\bm{\mathcal{E}}$ can also be solved.

To solve for $\bm{A}$, we break down the vector wave equation in \eqref{e2.45} into three separate scalar equations. Substituting these into \eqref{e2.45}, we obtain
\begin{eqnarray} \label{e2.49}
&&\hspace{-1cm} \nabla ^2 A_x + \beta^2 A_x = - \mu J_x, \notag\\
&&\hspace{-1cm} \nabla ^2 A_y + \beta^2 A_y = - \mu J_y, \notag\\
&&\hspace{-1cm} \nabla ^2 A_z + \beta^2 A_z = - \mu J_z ,
\end{eqnarray}
where $\beta = \omega^2 \mu \varepsilon = \frac{2\pi}{\lambda}$ is the phase constant of the transmitted wave, representing the phase change per unit distance.

The three equations in \eqref{e2.49} are structurally the same. Once we solve one of them, the solutions for the other two follow easily. To start, we find the solution for a point source. This solution, known as the unit impulse response, can be used to build a general solution by considering an arbitrary source as a combination of many point sources. The differential equation for a point source is
\begin{eqnarray} \label{e2.55}
\nabla^2 \psi + \beta ^2 \psi = - \delta(x) \delta(y) \delta(z),
\end{eqnarray}
where $\psi$ represents the response to a point source located at the origin, and $\delta(\cdot)$ is the unit impulse function. Although a point source is infinitesimally small, the current it represents has a specific direction. In practical problems, the point source models a small segment of current with a defined direction. If we consider the point source current to be directed along the $z$-axis, then $\psi = A_z$.

Since the point source only exists at the origin and is zero everywhere else, \eqref{e2.55} can be simplified to
\begin{eqnarray} \label{e2.57}
\nabla^2 \psi + \beta ^2 \psi = 0.
\end{eqnarray}
The two solutions to this equation are $\frac{e^{\pm j\beta |\bm{p}|}}{|\bm{p}|}$, where $|\bm{p}|$ is the distance from the point source to the observation point $\bm{p}$. These solutions represent waves moving outward and inward from the source, respectively. The physically relevant solution is the one that describes waves propagating outward from the point source. By determining the constant of proportionality, the solution for the point source becomes
\begin{eqnarray} \label{e2.58}
\psi = \frac{e^{-j \beta |\bm{p}|}}{4 \pi |\bm{p}|}.
\end{eqnarray}

For a current density directed along the $z$-axis, the resulting vector potential will also be $z$-directed. If we think of the source as a collection of point sources with a distribution given by $J_z$, the response $A_z$ can be found by summing the responses of each point source using \eqref{e2.58}. This summation is represented by an integral over the source volume $v$:
\begin{eqnarray} \label{e2.60}
A_z = \iiint \limits_{v} {\mu J_z \frac{e^{-j \beta |\bm{p}|}}{4 \pi |\bm{p}|} d v}.
\end{eqnarray}
Similar integrals apply for the $x$- and $y$-components. The overall solution is the sum of all components, expressed as
\begin{eqnarray} \label{e2.61}
\bm{A} = \iiint \limits_{v} {\mu \bm{J} \frac{e^{-j \beta |\bm{R}|}}{4 \pi |\bm{R}|} d v},
\end{eqnarray}
where $\bm{R}$ is the vector from the location of the point source $\bm{r}$ to the observation point $\bm{p}$, hence is given by $\bm{R}=\bm{p} - \bm{r}$. 
Therefore, we arrive at the solution to the vector wave equation \eqref{e2.45} as shown in \eqref{e2.61}.

In summary, when the current density $\bm{J}$ is known, the steps to find the induced field are:
\begin{enumerate}
    \item Use \eqref{e2.61} to find the magnetic vector potential $\bm{A}$; 
    \item Determine the magnetic field $\bm{\mathcal{H}}$ using \eqref{e2.36};  
    \item Apply \eqref{e:ME2} and consider that $\bm{J} = 0$ at the observation point $\bm{p}$ to derive the electric field as $\bm{\mathcal{E}} = \frac{1}{j \omega \varepsilon} \nabla \times \bm{\mathcal{H}}$.  
\end{enumerate}




For a half-wave dipole located at the origin and placed along the $z$-axis, the current density $\bm{J}$ distribution is given by
\begin{eqnarray} \label{e3.1}
J(z) = \sin{\left[ \beta\left( \frac{\lambda}{4} - |z| \right) \right] },
\end{eqnarray}
where $|z| \leq \frac{\lambda}{4}$.

The magnetic vector potential $\bm{A}$ can be derived from \eqref{e2.61} as follows:
\begin{eqnarray} \label{e2.101}
&&\hspace{-0.95cm} \bm{A} = \iiint \limits_{v} {\mu \bm{J} \frac{e^{-j \beta |\bm{R}|}}{4 \pi |\bm{R}|} d v} \\
&&\hspace{-0.6cm} \overset{(a)}{=} \mu \frac{e^{-j \beta |\bm{p}|}}{4 \pi |\bm{p}|}\iiint \limits_{v} { \bm{J} e^{-j \beta \frac{\bm{p}}{|\bm{p}|} \cdot \bm{r}} d v} \notag\\
&&\hspace{-0.6cm} \overset{(b)}{=} \bm{k}\mu \frac{e^{-j \beta |\bm{p}|}}{4 \pi |\bm{p}|}\int{ J(z) e^{-j \beta z  \cos{\theta} } d z }  \notag\\
&&\hspace{-0.6cm} \overset{(c)}{=} \bm{k}\mu \frac{e^{-j \beta |\bm{p}|}}{4 \pi |\bm{p}|}\int{  \sin{\left[ \beta\left( \frac{\lambda}{4} - |z| \right) \right] } e^{-j \beta z  \cos{\theta} } d z } \notag\\
&&\hspace{-0.6cm} \overset{(d)}{=} 2 \bm{k}\mu  \frac{e^{-j \beta |\bm{p}|}}{4 \pi |\bm{p}|} \frac{j \cos{\left(\frac{\pi}{2} \cos{\theta}\right)} + \sin{\left( \frac{\pi}{2} \cos{\theta} \right)} -\cos{\theta}}{j \beta \sin^2{\theta}}  ,\notag
\end{eqnarray}
where (a) follows from $|\bm{R}| = |\bm{p}| - \bm{r} \cdot \frac{\bm{p}}{|\bm{p}|} $ in the far field;
(b) follows because the polar angle of the observation point $\bm{p}$ is the emission angel of the transmitting antenna when the dipole is oriented along the $z$-axis;
(c) follows from \eqref{e3.1};
(d) follows from $\beta = \frac{2\pi}{\lambda}$, $\cos{x} = \frac{e^{jx} + e^{-jx}}{2}$, and $e^{j\frac{\pi}{2}} + e^{-j\frac{\pi}{2}} = 0$, $e^{j\frac{\pi}{2}} - e^{-j\frac{\pi}{2}} = 2j$.

Finally, the electric field induced at $\bm{p}$ can be derived as
\begin{eqnarray} \label{e3.2}
&&\hspace{-1cm} \bm{\mathcal{E}} = - j \omega \text{Re}\{A_\theta\} \bm{\vartheta} = j \omega \sin{\theta} \text{Re}\{A_z\} \bm{\vartheta} \notag\\
&&\hspace{-0.65cm} = j \omega \sin{\theta} 2 \mu  \frac{e^{-j \beta |\bm{p}|}}{4 \pi |\bm{p}|} \frac{\cos{\left(\frac{\pi}{2} \cos{\theta}\right)} }{ \beta \sin^2{\theta}} \bm{\vartheta}\notag\\
&&\hspace{-0.65cm} = j \omega \mu \frac{ 2 }{\beta}  \frac{e^{-j \beta |\bm{p}|}}{4 \pi |\bm{p}|} \frac{\cos{\left(\frac{\pi}{2} \cos{\theta}\right)} }{  \sin{\theta}} \bm{\vartheta}.
\end{eqnarray}

%% file: AppendixB.tex
Since the first multiplicative factor in \eqref{e:magnitude} is a constant, the behavior of the magnitude is determined by the second term:
$A \triangleq \frac{\cos{\left( \frac{\pi}{2} \cos{\theta^{(e)}_{\ell,\bm{p}}} \right)}}{\sin{\theta^{(e)}_{\ell,\bm{p}}}}$.
We analyze the derivative of $A$ with respect to (w.r.t) $\theta^{(e)}_{\ell,\bm{p}}$ to examine its variation. 

The derivative can be expressed as
\begin{eqnarray}\label{e:partialA}
&&\hspace{-1cm} \frac{\partial A}{\partial \theta^{(e)}_{\ell,\bm{p}}} 
= \frac{1}{\sin^2{\theta^{(e)}_{\ell,\bm{p}}}} \Bigg[ \left( -\frac{\pi}{4} - \frac{1}{2}\right) \cos{\left( \frac{\pi}{2} \cos{\theta^{(e)}_{\ell,\bm{p}}} + \theta^{(e)}_{\ell,\bm{p}}  \right)} \notag\\
&&\hspace{1 cm}  + \left( \frac{\pi}{4} - \frac{1}{2} \right) \cos{\left( \frac{\pi}{2} \cos{\theta^{(e)}_{\ell,\bm{p}}} - \theta^{(e)}_{\ell,\bm{p}} \right)} \Bigg],
\end{eqnarray} 
where $\theta^{(e)}_{\ell,\bm{p}} \in [0, \pi]$. To simplify the notations, we define $A_1 \triangleq \frac{\pi}{2} \cos{\theta^{(e)}_{\ell,\bm{p}}} + \theta^{(e)}_{\ell,\bm{p}}$ and $A_2 \triangleq \frac{\pi}{2} \cos{\theta^{(e)}_{\ell,\bm{p}}} - \theta^{(e)}_{\ell,\bm{p}}$.

When $ \theta^{(e)}_{\ell,\bm{p}} \in (0, \frac{\pi}{2}) $, we have
$$ \frac{\partial A_1}{\partial \theta^{(e)}_{\ell,\bm{p}}} = -\frac{\pi}{2} \sin{\theta^{(e)}_{\ell,\bm{p}}} + 1 .$$
Thus, $\frac{\partial A_1}{\partial \theta^{(e)}_{\ell,\bm{p}}}$ is a monotonically decreasing function of $ \theta^{(e)}_{\ell,\bm{p}} $ from 1 to $- \frac{\pi}{2} + 1 < 0$. The minimum value of $ A_1=\frac{\pi}{2} $ occurs at $\theta^{(e)}_{\ell,\bm{p}} = 0$ or $\theta^{(e)}_{\ell,\bm{p}} = \frac{\pi}{2}$.
Since $$A_1 \leq \frac{\pi}{2} + \theta^{(e)}_{\ell,\bm{p}} < \pi,$$ we have $ \frac{\pi}{2} < A_1 < \pi $ for $ \theta^{(e)}_{\ell,\bm{p}} \in (0, \frac{\pi}{2}) $.
Therefore, $ \cos{\left( A_1 \right)} < 0$, and $$ \left( -\frac{\pi}{4} - \frac{1}{2} \right) \cos{\left( A_1 \right)} > 0 .$$

On the other hand, we have
$$ \frac{\partial A_2}{\partial \theta^{(e)}_{\ell,\bm{p}}} = -\frac{\pi}{2} \sin{\theta^{(e)}_{\ell,\bm{p}}} - 1 < 0.$$
$A_2 $ is a monotonically decreasing function when $ \theta^{(e)}_{\ell,\bm{p}} \in (0, \frac{\pi}{2}) $, with $ A_2(0) = \frac{\pi}{2} $ and $ A_2(\frac{\pi}{2}) = -\frac{\pi}{2} $. Therefore, $ |A_2| < \frac{\pi}{2} $, implying that $ \cos\left( \frac{\pi}{2} \cos \theta - \theta \right) > 0 $, and thus
$$\left( \frac{\pi}{4} - \frac{1}{2} \right) \cos{\left( A_2 \right)} > 0.$$

Additionally, $ \sin^2{\theta^{(e)}_{\ell,\bm{p}}} > 0 $ when $ \theta^{(e)}_{\ell,\bm{p}} \in (0, \frac{\pi}{2}) $.
Overall, we have $\frac{\partial A}{\partial \theta^{(e)}_{\ell,\bm{p}}}  > 0 $ when $ \theta^{(e)}_{\ell,\bm{p}} \in (0, \frac{\pi}{2}) $. 

Similarly, it can be shown that $ \frac{\partial A}{\partial \theta^{(e)}_{\ell,\bm{p}}} < 0 $ when $ \theta \in (\frac{\pi}{2}, \pi) $.
Since $ \frac{\partial A}{\partial \theta^{(e)}_{\ell,\bm{p}}}  \big|_{\theta = \frac{\pi}{2}} = 0 $, both $ A$ and the magnitude $|\mathcal{E}_{\ell,\bm{p}}|$ reach their maximum when $ \theta^{(e)}_{\ell,\bm{p}} = \frac{\pi}{2}$, proving the corollary.

%% file: AppendixC.tex
In partiuclar, we introduce small Gaussian perturbations are added to the user locations and orientations:
\begin{itemize}
    \item The unit error in $\bm{p}_t$ is 0.1 $\lambda$;
    \item The unit error in $\bm{p}_r$ is 1 cm;
    \item The unit error in $\bm{n}_t$ and $\bm{n}_r$ is 1$^\circ$.
\end{itemize}

The simulation results are shown in Fig.~\ref{fig:3_1}, where the horizontal axis represents the level of unit errors, and the vertical axis indicates the relative SNR gain. A value of zero on the vertical axis corresponds to no gain from the optimization process.  
As seen in the figure, the PAMA system is robust to directional errors of a few degrees at both the transmitter and receiver, as well as centimeter-level position errors at the receiver. On the other hand, it is more sensitive to transmitter position errors on the order of the wavelength, which can cause a decline in system performance. Despite this, the system can still achieve performance gains in the presence of estimation errors, demonstrating the reliability of the PAMA system.

\begin{figure}[!tb]
    \centering
    \includegraphics[width=0.9\linewidth]{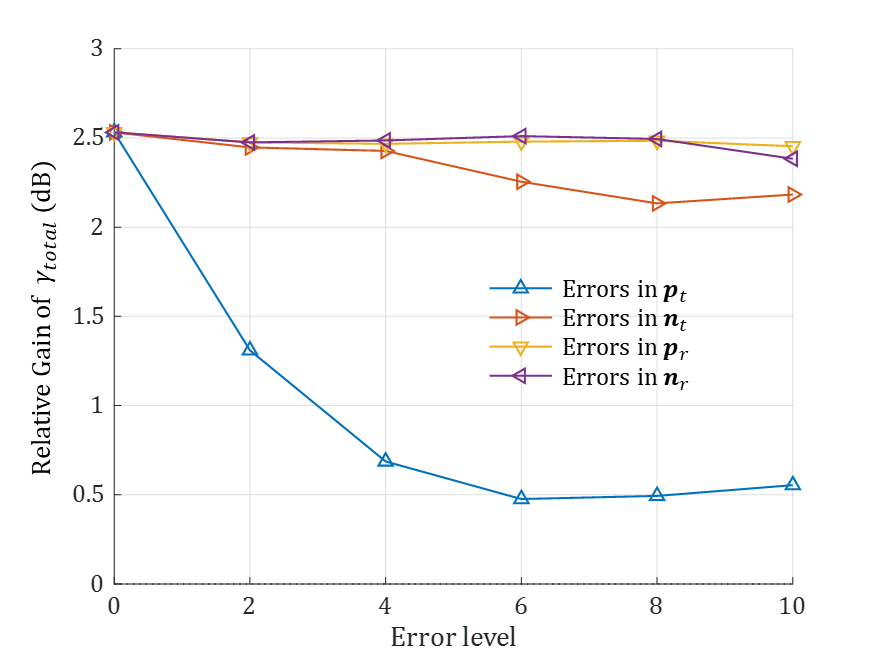}
    \caption{  The relative gain of equivalent total SNR versus the error in different parameters.}
    \label{fig:3_1}
\end{figure}

%% file: AppendixD.tex
We have analyzed the optimization approaches of Generalized Minimum Mean Square Error (GMMSE) and Generalized Power Iteration-based Precoding (GPIP) in the PAMA scenario and compared them with the original ``ZF + Waterfilling'' method used in our paper.

\textbf{GMMSE}:
By establishing the equivalence between maximizing the weighted sum-rate and minimizing the weighted MMSE, the optimal beamformer can be obtained iteratively. The key steps are as follows:

\begin{enumerate}
    \item The beamforming matrix is initialized as $\bm{B}_k = b \bm{H}_k^H$, where the scaling facto $b$ ensures compliance with the power constraint.
    \item MMSE-receive filter is obtained as
    \begin{eqnarray}
        \bm{A}_k^{MMSE} &=& \arg \underset{\bm{A}_k}{\min} {\mathbb{E} \left[ || \bm{A}_k \bm{y}_k - \bm{d}_k ||^2 \right]} \\
        &=& \bm{B}_k^H \bm{H}_k^H \left( \bm{H}_k \bm{B}_k \bm{B}_k^H \bm{H}_k^H + \bm{R}_{\tilde{v}_k\tilde{v}_k}\right)^{-1},\notag
    \end{eqnarray}
    where $\bm{R}_{\tilde{v}_k\tilde{v}_k} = \bm{I} + \sum_{i = 1 , i \neq k}^K {\bm{H}_k \bm{B}_i \bm{B}_i^H \bm{H}_k^H}$ denotes the effective noise covariance matrix at user $k$.
    \item  Given the MMSE-receive filteris, the MSE-matrix for user $k$ can be written as 
    
    \begin{eqnarray}
        \bm{E}_k &=& {\mathbb{E} \left[ || \bm{A}_k^{MMSE} \bm{y}_k - \bm{d}_k ||^2 \right]} \notag\\
        &=& \left( \bm{I} +  \bm{B}_k^H \bm{H}_k^H \bm{R}_{\tilde{v}_k\tilde{v}_k}^{-1} \bm{H}_k \bm{B}_k\right)^{-1},\notag
    \end{eqnarray}
    and $\bm{W}_k = u_k\bm{E}_k^{-1}$.
    \item The WMMSE transmit filter structure is 
    $$\bar{\bm{B}} \!\!=\!\! \left(\!\! \bm{H}^{\!H} \! \bm{A}^{\!H} \!  \bm{W} \! \bm{A} \! \bm{H} \! + \! \frac{\text{Tr} \! \left(\bm{W} \!  \bm{A} \bm{A}^{\!H} \!  \right)}{E_{tx}} \bm{I}_L\right)^{-1} \!\!\!\bm{H}^{\!H}\! \bm{A}^{\!H}\! \bm{W},$$ where $\bm{W} = \text{diag}{\left\{ \bm{W}_1, \dots ,  \bm{W}_K \right\}}$, $\bm{A} = \text{diag}{\left\{ \bm{A}_1, \dots ,  \bm{A}_K \right\}}$, and $\bm{H} = \left[ \bm{H}_1^\top, \dots , \bm{H}_K^\top \right]^\top$.
    \item $\bm{B}^{\text{MMSE}} = b \bar{\bm{B}}$, where $b = \sqrt{\frac{E_{tx}}{\text{Tr} \left( \bar{\bm{B}} \bar{\bm{B}}^H \right)}}$ is a gain factor which scales the signal  so as to satisfy the transmit power constraint.
\end{enumerate}

All aforementioned parameters follow the definitions established in \cite{shi2011iteratively}, with Algorithm Steps 2-5 being iteratively executed until convergence is achieved. 
Given that the beamformer matrix $\bm{B}$ corresponds to the parameterization in our framework through the relation $\bm{B} = \bm{W}\bm{P}$, the optimal solution $\bm{B}^{\text{MMSE}}$ obtained through this procedure can be factorized to derive the normalized optimal beamforming vectors $\{\bm{w}_k\}_{k=1}^K$ with their optimal power allocation coefficients $\{P_k\}_{k=1}^K$ for PAMA systems.

\textbf{GPIP}:
The weighted sum-rate maximization problem can be reformulated as an equivalent matrix optimization problem, where maximizing the sum-rate reduces to finding the dominant singular value. The corresponding singular vectors yield the optimal beamformer. The procedure is as follows:
\begin{enumerate}
    \item Two intermediate variables are initialized as:
    $$\bm{A}_k = \begin{bmatrix}
        \bm{h}_k \bm{h}_k^H + \bm{\Phi}_k & \!\!\!\dots\!\!\! & 0 \\
        \dots  & \!\!\!\dots\!\!\!  & \dots \\
        0  & \!\!\!\dots\!\!\!  &  \bm{h}_k \bm{h}_k^H + \bm{\Phi}_k \\
    \end{bmatrix} + \frac{\sigma_k^2}{P} \bm{I}_{LK},$$
    $$\bm{B}_k = \bm{A}_k 
 - \begin{bmatrix}
        0 & \dots & 0 & \dots  & 0 \\
        \dots  & \dots  & \dots  & \dots  & \dots  \\
        0  & \dots  &  \bm{h}_k \bm{h}_k^H & \dots  & 0  \\
        \dots  & \dots  & \dots  & \dots  & \dots  \\
        0  & \dots  & 0  & \dots  &  0  \\
    \end{bmatrix}.$$
    \item Through differentiation with respect to the beamformer, we obtain two groups of coefficients    
    \begin{eqnarray}
    c_i\left( \bm{f}^{(m - 1)} \right) = \left( w_i \left( \bm{f}^{(m - 1)} \right)^H \bm{A}_i \bm{f}^{(m - 1)}  \right)^{w_i - 1} \notag\\
    \prod_{k \neq i}^K{ \left( \left( \bm{f}^{(m - 1)} \right)^H \bm{A}_k \bm{f}^{(m - 1)}  \right)^{w_i}},\notag
    \end{eqnarray}
    \begin{eqnarray}d_i\left( \bm{f}^{(m - 1)} \right) = \left( w_i \left( \bm{f}^{(m - 1)} \right)^H \bm{B}_i \bm{f}^{(m - 1)}  \right)^{w_i - 1}\notag\\
    \prod_{k \neq i}^K{ \left( \left( \bm{f}^{(m - 1)} \right)^H \bm{B}_k \bm{f}^{(m - 1)}  \right)^{w_i}}.\notag
    \end{eqnarray}
    \item Construct two intermediate variables as
    $$\bar{\bm{A}}\left( \bm{f}^{(m - 1)} \right) = \sum_{i = 1}^K {c_i \left( \bm{f}^{(m - 1)} \right) \bm{A}_i},$$ $$\bar{\bm{B}}\left( \bm{f}^{(m - 1)} \right) = \sum_{i = 1}^K {d_i \left( \bm{f}^{(m - 1)} \right) \bm{B}_i}.$$
    \item The iterative beamformer update yields $\bm{f}^{(m)} = \left[ \bar{\bm{B}}\left( \bm{f}^{(m - 1)} \right) \right]^{-1} \bar{\bm{A}}\left( \bm{f}^{(m - 1)} \right) \bm{f}^{(m - 1)}$
    \item After applying power normalization, we obtain the final beamformer $\bm{f}^{(m)} = \frac{\bm{f}^{(m)} }{|\bm{f}^{(m)} |}.$
\end{enumerate}

The iterative procedure repeats Steps 2-5 until convergence is achieved. All parameters follow the definitions established in \cite{choi2019joint}, where the beamforming vector is constructed as $\bm{f} = [\bm{f}^\top_1, \dots, \bm{f}^\top_K]^\top$. This formulation maintains the fundamental relationship $\bm{f}_k = P_k\bm{w}_k$, with both the power allocation coefficients $P_k$ and beamforming vectors $\bm{w}_k$ being essential components for PAMA system implementation.

\begin{figure}[!tb]
    \centering
    \includegraphics[width=0.9\linewidth]{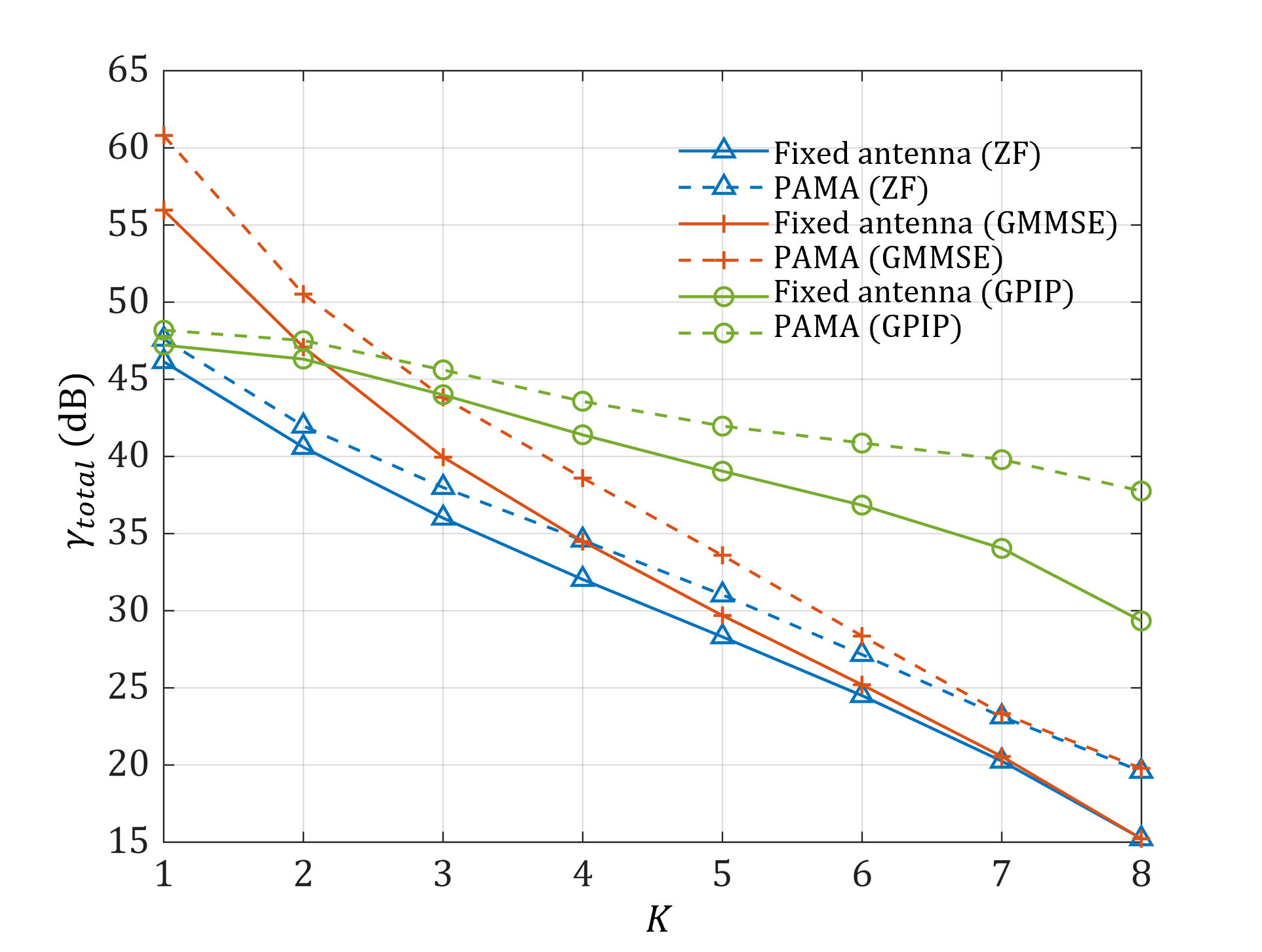}
    \caption{The equivalent total SNR $\gamma_{\text{total}}$ versus the number of users $K$ for three different beamforming algorithms: ZF, GMMSE, and GPIP.}
    \label{fig:1_2}
\end{figure}

Fig.~\ref{fig:1_2} presents a comparative analysis of the total equivalent SNR across the aforementioned algorithms, including ZF, GMMSE, and GPIP, for both fixed antenna and PAMA systems. As observed:
\begin{itemize}
    \item Regardless of the optimization algorithm employed, the PAMA system consistently outperforms its fixed-antenna counterpart, thanks to the combined gains from antenna mobility and polarization matching.
    \item As the reviewer rightly noted, employing more advanced detection algorithms, such as GMMSE and GPIP, can lead to further SNR improvements over the ZF baseline previously used in the main manuscript.
\end{itemize}
